\documentclass[11pt]{article}
\usepackage{titling}
\usepackage[english]{babel}
\usepackage{blindtext}
\usepackage{tikz}
\usepackage{amsmath,amsthm,amsfonts,bm}
\usepackage{stmaryrd}
\usepackage{xfrac}
\usepackage{url}
\usepackage{multirow}
\usepackage{subcaption}
\usepackage[margin=1in]{geometry}
\usepackage[cal=pxtx]{mathalfa}
\usepackage{cleveref}
\usepackage{wrapfig}
\usepackage[]{algorithm2e}
\usepackage{thmtools}
\usepackage{thm-restate}
\usepackage{authblk}
\usepackage[normalem]{ulem}

\newtheorem{thm}{Theorem}
\newtheorem{lem}{Lemma}

\newtheorem{prop}{Proposition}
\newtheorem{clm}{Claim}
\theoremstyle{definition}
\newtheorem{dfn}{Definition}

\newcommand{\eps}{{\varepsilon}}
\newcommand{\lowerbd}{L^*}

\newcommand{\gpi}{G_{\bm{\pi}}}
\newcommand{\evirt}{E_{\text{virt}}}
\newcommand{\ephys}{E_{\text{phys}}}

\newcommand{\wphys}{W_{\text{phys}}}

\newcommand{\rplus}{[0,\infty)}
\newcommand{\pths}{\mathcal{P}}
\newcommand{\timeline}[1]{\left\llbracket #1 \right\rrbracket}
\newcommand{\floor}[1]{\left\lfloor #1 \right\rfloor}

\newcommand{\traversing}{F}

\newcommand{\dprime}{D^\prime}
\newcommand{\rprime}{R^\prime}
\newcommand{\datob}{D^\prime_{a\rightarrow b}}
\newcommand{\dbtoc}{D^\prime_{b\rightarrow c}}
\newcommand{\dall}{D^{ALL}}

\newcommand{\rout}{\mathcal{R}}

\renewcommand{\ell}{h}

\newenvironment{lparray}%
{\begingroup  \begin{array}{l@{\hspace{8mm}}l@{\hspace{8mm}}l}}%
{\end{array} \endgroup}

\date{}

\begin{document}

\begingroup
\let\clearpage\relax

\title{Optimal Oblivious Reconfigurable Networks}

 \author[1]{Daniel Amir\thanks{Author order was randomized with students placed before professors.}}
 \author[1]{Tegan Wilson}
 \author[2]{Vishal Shrivastav}
 \author[1]{Hakim Weatherspoon}
 \author[1]{Robert Kleinberg}
 \author[1]{Rachit Agarwal}

 \affil[1]{Cornell University}
 \affil[2]{Purdue University}

\begin{titlingpage}
\maketitle
\begin{abstract}
Oblivious routing has a long history in both the theory and practice of networking.
In this work we initiate the formal study of oblivious routing in the context of reconfigurable networks, a new architecture that has recently come to the fore in datacenter networking. These networks allow a rapidly
changing bounded-degree pattern of interconnections
between nodes, but the network topology
and the selection of routing paths must both
be oblivious to the traffic demand matrix.
Our focus is
on the trade-off between maximizing throughput
and minimizing latency in these networks.
For every constant throughput rate, we
characterize (up to a constant factor)
the minimum latency achievable by an
oblivious reconfigurable network
design that satisfies the given throughput
guarantee. The trade-off between these two
objectives turns out to be surprisingly subtle:
the curve depicting it has an unexpected
scalloped shape reflecting the fact that
load-balancing becomes more difficult when
the average length of routing paths is not
an integer because equalizing all the path
lengths is not possible. The proof of our
lower bound uses LP duality to verify that
Valiant load balancing is the most efficient
oblivious routing scheme when used in
combination with an optimally-designed
reconfigurable network topology.
The proof of our upper bound uses an
algebraic construction in which the
network nodes are identified with
vectors over a finite field, the
network topology is described by
either the elementary basis or
a sequence of Vandermonde matrices,
and routing paths are constructed by
selecting columns of these matrices
to yield the appropriate mixture
of path lengths within the shortest
possible time interval.
\end{abstract}
\end{titlingpage}

\section{Introduction}

Oblivious routing has a long history in both the theory and practice of networking.
By design, an oblivious routing scheme forwards data along a fixed path (or distribution
over paths) designed to provide good performance across a wide range of possible
traffic demand matrices.
Past theoretical work on oblivious routing schemes
focused on their ability to approximate the congestion of the optimal
multicommodity flow, culminating in R\"{a}cke's discovery~\cite{raecke08} of
oblivious routing schemes for general networks that are guaranteed
to approximate the optimum congestion within a logarithmic factor
in the worst case.
However, thus far, oblivious routing has only been studied in the context of static
networks, where the edges in the network are fixed at the
beginning and do not change over time.
Recent advances in datacenter network
architecture~\cite{c-through,helios,microsecond-circuit-switching-dc,reactor-circuit-switching,ProjecToR-free-space-optics,rotornet,opera,shoal} have brought \textit{reconfigurable networks}
to the fore. A reconfigurable network is defined as a $d$-regular network with $N$ nodes (or hosts) where the
edges (or links) between the nodes can be reconfigured (or rearranged)
very rapidly over time.
Early designs of reconfigurable networks for datacenters~\cite{c-through,helios,reactor-circuit-switching} relied on predictable traffic demand matrices to
choose optimal edge configurations and routes for sending data between nodes.
However, more recent works~\cite{rotornet,opera,shoal} in this space have made a case that
traffic demand matrices in datacenters are highly unpredictable and change at very
fine time granularities,
making it challenging, if not impossible, to accurately track the
demand matrix at any given time.
To overcome this fundamental challenge, recent works
have advocated for edge configuration and route selection mechanisms that are
\textit{oblivious} to traffic demand matrices.
In this paper, we make the first attempt to formally study
the problem of oblivious routing
in the novel context of reconfigurable networks.

    There are two key objectives that
    oblivious reconfigurable networks
    must aim to optimize.
    First, since it is costly to overprovision networks (especially for modern high-bandwidth links),
    datacenter network operators aim for extremely high throughput, utilizing a large \textit{constant} factor
    of the available network capacity at all times if possible.
    At the same time, it is desirable to minimize latency,
    the worst-case delay between when a packet arrives to the network
    and when it reaches its destination.
    Thus, there is a vital need to understand
    oblivious network designs for reconfigurable networks
    that
    guarantee high throughput and low maximum latency.

    The objectives of maximizing throughput and minimizing latency
    in reconfigurable networks are in conflict:
    due to degree constraints most nodes cannot be connected by a
    direct link at all times, so one has to either use
    indirect paths, which comes at the expense
    of throughput, or settle for higher latency while
    waiting for reconfigurations to yield a more direct path.
    Since different deployments (and applications) may necessitate different tradeoffs between
    these two conflicting objectives, the main question that our work investigates is the following:
\begin{quotation}
  \noindent {\em For every throughput rate $r$, what is the minimum latency achievable by an oblivious reconfigurable network design that guarantees throughput $r$?}
\end{quotation}
We fully resolve this question to within a constant factor\footnote{\label{fn:prioritizing-throughput}
  One could, of course, ask the transposed question: {\em for every
  latency bound $L$, what is the maximum guaranteed throughput
  rate achievable by an oblivious routing scheme with maximum latency $L$?}
  Our work also resolves this question, not only to within a constant
  factor, but up to an additive error that tends to zero as $N \to \infty$.
  As noted below in \Cref{sec:tech}, optimizing throughput to within a factor of two, subject
  to a latency bound, is much easier than optimizing latency to within
  a constant factor subject to a throughput bound. The importance of the
  latter optimization problem, {\em i.e.~}our main question,
  is justified by the high cost of overprovisioning networks:
  due to the cost of overprovisioning,
  datacenter network operators tend to be much less tolerant of
  suboptimal throughput than of suboptimal latency.}
for $d$-regular reconfigurable networks, except when $d$ is very large --- bounded below by a constant power of $N$ (the number of nodes in the network). That is, for every constant rate $r$, we identify a lower bound $\frac{1}{d}\lowerbd(r,N)$ such that any $N$-node $d$-regular reconfigurable network guaranteeing throughput $r$ must have maximum latency bounded below by $\frac{1}{d}\lowerbd(r,N)$. Complementing this lower bound, we design oblivious networking schemes that guarantee throughput $r$ and have maximum latency bounded by $O(\frac{1}{d}\lowerbd(r,N))$, for every constant
$r \in (0,\frac12], d \in \mathbb{N},$
and infinitely many $N$. (For $r > \frac12 + o(1)$,
we show in \Cref{app:general-throughput} that
it is impossible for
oblivious network designs to guarantee throughput $r$.)

The shape of the optimal tradeoff curve between throughput and latency is quite surprising. Figure~\ref{fig:lowerbound-curve} depicts the curve for $N=10^9$ and $d=1$; the $x$-axis measures the inverse throughput, $1/r$, while the $y$-axis (in log scale) measures maximum latency. The curve is scallop-shaped, with particularly favorable tradeoffs occurring when $1/r$ is an even integer. Between even-integer values of $1/r$, the maximum latency improves slowly at first, then precipitously as $1/r$ approaches the next even integer.
The proof of our main result explains these key features of the tradeoff curve: its non-convexity, the special role played by even integer values of $1/r$, and the steep but continuous improvement in $\lowerbd(r,N,d)$ as $1/r$ approaches the next even integer. In \Cref{sec:tech} below we sketch the intuitions that account for these features. Before doing so, we pause to explain more fully our model and notation.

\begin{figure}[h]
 \centering
 \includegraphics[width=0.6\textwidth]{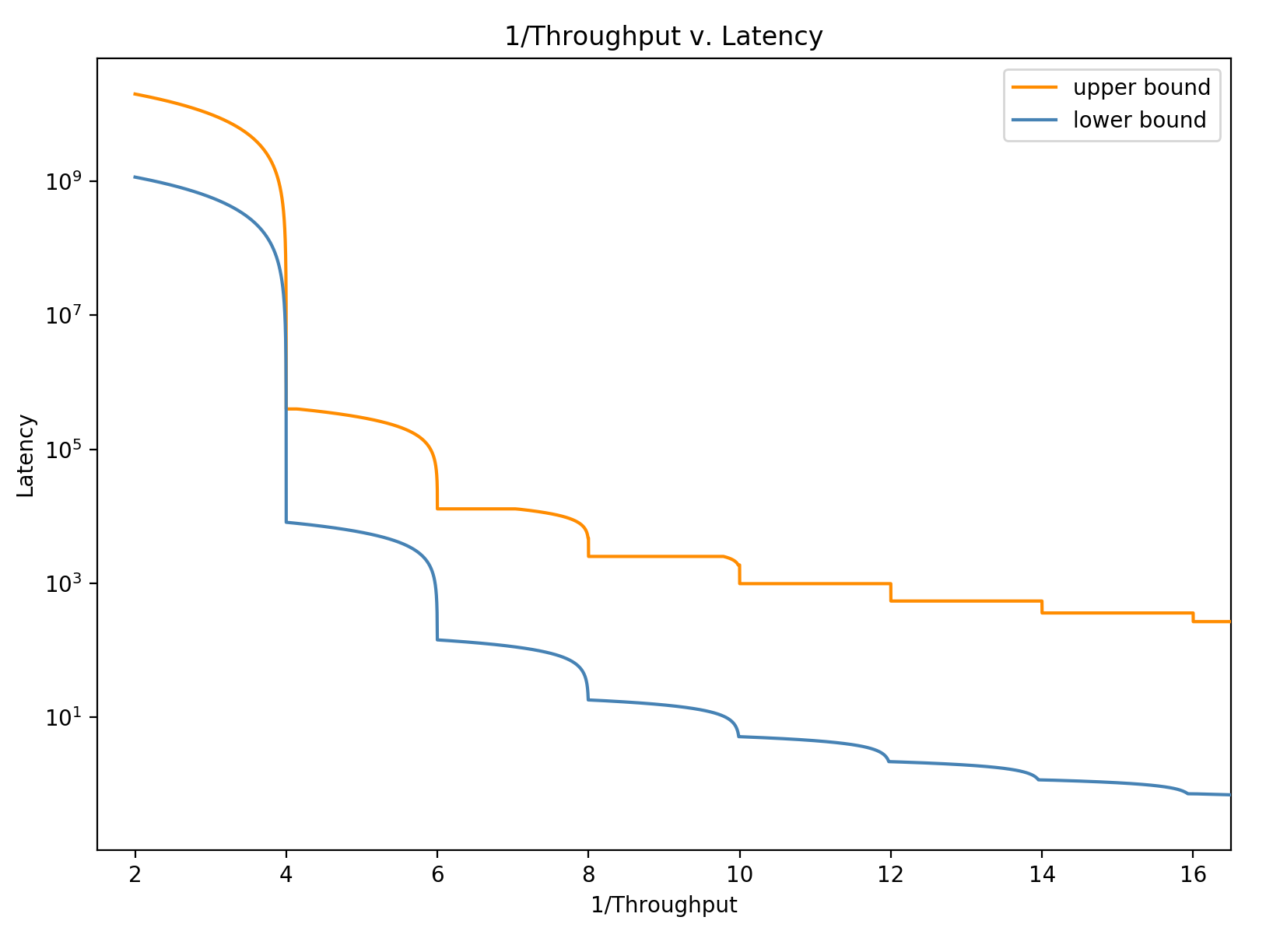}
 \caption{A plot of the upper and lower bounds for the latency of an ORN containing $10^9$ nodes that can guarantee a given throughput.}
 \label{fig:lowerbound-curve}
\end{figure}

\subsection{Our Model and Results}
\label{sec:model-and-results}


Our model of oblivious reconfigurable networking is inspired by the circuit-switched network designs popularized by works such as~\cite{rotornet,opera,shoal}. These are networks composed of a fixed set of $N$ nodes, with a switching fabric that allows a time-varying pattern of links providing connectivity between node pairs.
A network design in our model is specified by two ingredients: a connection schedule and an oblivious routing scheme.
The connection schedule designates which node pairs are connected in each timeslot.
This can be visualized in the form of a virtual topology: a layered directed graph (with layers corresponding to timeslots) that encodes the paths that network traffic can take over time.
The oblivious routing scheme designates, for each source-destination pair $(a,b)$ and timeslot $t$, a probability distribution over routing paths used to forward traffic with destination $b$ that originates at $a$ in timeslot $t$.
A routing path is specified by the sequence of edges in the virtual topology that compose the path. We call the combination of a connection schedule and an oblivious routing scheme an {\em oblivious reconfigurable network (ORN) design}.

We evaluate ORN designs according to two quantities: maximum latency ($L$) and guaranteed throughput ($r$). Latency of a path measures the difference between the timeslots when it starts and ends, and an ORN design with maximum latency $L$ uses no routing paths of latency greater than $L$. The definition of guaranteed throughput is more subtle.
First, we model demand using a function that specifies, for each source-destination pair and each timeslot, the amount of flow with that source and destination originating at that time.
We say an ORN design {\em guarantees throughput $r$} if the routing scheme is guaranteed not to exceed the capacity of any link, whenever the demand satisfies the property that the total amount of demand originating at any source, or bound for any destination, never exceeds $r$ at any timeslot. Our main result can now be stated in the following form.

\begin{thm} \label{thm:lb-informal}
  Consider any constant $r \in (0,\frac12].$ Let $(h,\eps)$
  to be the unique solution in $\mathbb{N} \times (0,1]$ to the
  equation $\frac{1}{2r} = h + 1 - \eps$, and let $\lowerbd(r,N)$
  be the function
  \[
    \lowerbd(r,N) = h \left( N^{1/(h+1)} + (\eps N)^{1/h} \right) .
  \]
  For every $N > 1$ and
  every ORN design on $N$ nodes that guarantees throughput $r$,
  the maximum latency is at least $\Omega(\lowerbd(r,N))$.
  Furthermore for infinitely many $N$
  there exists an ORN design on $N$ nodes that guarantees
  throughput $r$ and whose maximum latency is $O(\lowerbd(r,N))$.
\end{thm}

\subsection{Techniques}
\label{sec:tech}

To begin reasoning about the latency-throughput tradeoff in ORNs, note that for any node in the virtual topology, the number of distinct routing paths originating at that node whose latency is at most $L$ and which contain $p$ physical edges is $\binom{L}{p}$.
Hence, in order for a node to be able to reach a majority of other nodes within $L$ timeslots using at most $\ell$ physical links, we must have the inequality $\sum_{p=0}^{\ell} \binom{L}{p} \ge N/2$.
A simple calculation verifies that this inequality implies $L = \Omega \left( \ell N^{1/\ell} \right)$.
A routing scheme in which the routing path between a random source and a random destination contains $\ell$ physical links, on average, cannot guarantee throughput greater than $1/\ell$. This suggests a latency-throughput relationship of the form $L = \Omega \left( \frac{1}{r} N^{r} \right)$.
This lower bound can be made rigorous with a little bit of work, but it differs from the tight bound asserted in \Cref{thm:lb-informal} in two significant ways.

\begin{enumerate}
  \item Whereas $\frac1r N^{r}$ is a smooth convex function of $r > 0$, the
    function $\lowerbd(r,N)$ is non-smooth and non-convex; when plotted as a
    function of $1/r$ it exhibits a scalloped shape with cusps at even integer
    values of $1/r$.
  \item The exponent of $N$ in the function $\lowerbd(r,N)$ is approximately
    $2r$ rather than $r$. In other words, the na\"{i}ve bound
    $L \ge \frac1r N^{r}$ is tight up to a factor of 2 in terms
    of throughput, but off by a factor of about $N^{r}$ in terms
    of latency.
      (As remarked in \Cref{fn:prioritizing-throughput},
      sacrificing a factor of 2 in throughput
      is typically regarded by network operators
      as much more costly than
      sacrificing a constant factor in latency.)
\end{enumerate}

The first of these differences is explained by a refinement of the counting argument at the start of this section. In order to guarantee throughput $r$, the average number of physical hops on the routing paths used (under any traffic demands with at most $r$ units of flow based at any source or destination) must be at most $1/r$. However, the number of physical hops in any path must be an integer. Thus, if $1/r$ is not an integer, at least a constant fraction of routing paths must have $\lfloor 1/r \rfloor$ physical hops or fewer. Subject to any upper bound on latency, paths with a limited number of physical hops are much less numerous than those with a larger number of physical hops, so the requirement to use a large number of distinct paths with $\lfloor 1/r \rfloor$ or fewer physical hops places a significantly stricter lower bound on maximum latency, leading to the non-convex shape with regularly spaced cusps depicted in \Cref{fig:lowerbound-curve}.

To give intuition for the factor-two difference in throughput between the na\"{i}ve lower bound and the true function $\lowerbd(r,N)$, it is useful to recall {\em Valiant load balancing (VLB)}, an ingredient in many of the earliest and most practical oblivious routing schemes. VLB constructs a random path from source $s$ to destination $t$ by choosing a random intermediate node, $r$, and concatenating  minimum-cost paths from $s$ to $r$ and from $r$ to $t$. This inflates the number of physical hops used in routing paths by a factor of two, but is beneficial because it prevents congestion under worst-case demands. The fact that the exponent of $N$ in $\lowerbd(r,N)$ is approximately $2r$ rather than $r$ can be interpreted as confirming that the factor-two inflation due to VLB is unavoidable, for oblivious routing schemes that guarantee throughput $r$. To prove this fact, we formulate optimal oblivious routing for a given virtual topology as a linear program and interpret the dual variables as endpoint-specific edge costs that can be summed to ascribe a cost to every path connecting a given pair of endpoints. We prove that, regardless of the virtual topology, one can always design a carefully-constructed dual solution that penalizes paths containing a large number of physical hops, and doubly penalizes physical hops that are too close to both endpoints. Paths that avoid the double penalty must use twice as many physical hops as minimum-cost paths, exactly as in VLB routing.
The most delicate part of the proof is the verification that the dual solution is feasible, which requires carefully bounding the number of nodes reachable from any source within a given cost budget.

To prove that the lower bound $\lowerbd(r,N)$ is tight, we need to construct an ORN design that matches the bound up to a constant factor.
Our design is easiest to describe when $r = \frac{1}{2h}$ and $N = n^h$ for positive integer $h$ and prime number $n$.
In that case, we use a design that we call the {\em Elementary Basis Scheme} (EBS) which identifies the set of $N$ nodes with elements of the group $(\mathbb{Z}/(n))^h$.
Let $\mathbf{e}$ be the elementary basis consisting of the columns of the $h \times h$ identity matrix.
EBS uses a connection schedule whose timeslots cycle through the nonzero scalar multiples of elements of $Y$.
In a timeslot devoted to $s \cdot \mathbf{e}_i,$ the network is configured to allow each node $x$ to send to $x + s \cdot \mathbf{e}_i$. Over the course of one complete cycle, any two nodes can be connected by a ``direct path'' consisting of $h$ physical hops (or fewer) that modify the coordinates of the source node one by one until they match the coordinates of the destination.
The EBS routing scheme constructs a random path connecting a given source and destination using VLB: it chooses a random intermediate node and concatenates two ``semi-paths'': the direct paths from the source to the intermediate node and from the intermediate node to the destination.

To generalize this design to all non-integer values of
$\frac{1}{2r}$, we need to enhance EBS so that
a constant fraction of semi-paths use $h$ physical
hops and a constant fraction use $h+1$ physical hops.
This necessitates a modified ORN design that we call
the {\em Vandermonde Basis Scheme} (VBS).
Assume
$r = h+1-\eps$ for $h \in \mathbb{N}, 0 < \eps < 1,$
and that $N = n^{h+1}$ for prime $n$, so that the
nodes can be identified with the vector space
$\mathbb{F}_n^{h+1}$. Instead of one basis
corresponding to the identity matrix,
we now use a sequence of distinct bases each
corresponding to a different Vandermonde matrix.
In addition to the single-basis semi-paths (which
now constitute $h+1$
physical hops), this enables the creation of
``hop-efficient'' semi-paths composed of $h$
physical hops belonging to two or more of the
Vandermonde matrices in the sequence.
Hop-efficient semi-paths
have higher latency than direct paths, but we
opportunistically use only the ones with lowest
latency to connect a subset of terminal
pairs, joining the remaining pairs with
direct semi-paths. A full routing path is
then defined to be the concatenation of two
random semi-paths, as before. Proving that
the routing scheme guarantees throughput
$r$ boils down to quantifying, for each physical
edge $e$, the net effect of shifting load from
direct paths that use $e$ to hop-efficient paths
that avoid $e$ and vice-versa. The relevant sets
of paths in this calculation can be parameterized
by unions of affine subspaces of $\mathbb{F}_n^{h+1}$,
and the use of Vandermonde matrices in the
connection schedule gives us control over the
dimensions of intersections of these subspaces,
and thus over the size of their union.
%

\subsection{Related work}





{\bf Oblivious routing in general networks:}
R\"{a}cke's seminal 2002 paper \cite{raecke02} proved the existence of $\operatorname{polylog}(n)$-competitive oblivious routing schemes in general networks.
Subsequent work improved the competitive ratio \cite{harrel03} and devised polynomial-time algorithms for computing an oblivious routing scheme that meets this bound \cite{bienk03,harrel03,azar03}.
R\"{a}cke's 2008 paper \cite{raecke08} yielded an $O(\log n)$-competitive oblivious routing scheme, computed by a fast, simple algorithm based on multiplicative weights and FRT's randomized approximation of general metric spaces by tree metrics \cite{frt04}.
The effectiveness of R\"{a}cke's 2008 routing scheme for wide-area traffic engineering in practice was demonstrated in \cite{applegate04,smore18}.
Additionally, Gupta, Hajiaghayi, and R\"{a}cke \cite{gupta06} show a $\operatorname{polylog}(n)$ competitive ratio for routing schemes oblivious to both traffic and the cost functions associated with each edge.
While these works achieve excellent congestion minimization over general networks, they do not specifically consider throughput or latency, and do not attempt to co-design the network with their routing scheme.

With respect to bounding the throughput of oblivious routing schemes, Hajiaghayi, Kleinberg, Leighton, and R\"{a}cke \cite{hajia06} prove a lower bound of $\Omega(\frac{\log n}{\log \log n})$ on the competitive ratio in general networks. 
However, their definition of throughput differs from ours; they simply mean the combined flow rate delivered to all sender-receiver pairs.
With respect to latency, the competitive ratio of average latency of oblivious routing over general networks is analyzed by \cite{harsha08}.
Their model of latency differs from ours; they assign resistance values to each edge, and they only provide an oblivious routing scheme achieving the $O(\log(N))$-competitive ratio when routing to a single target.

{\bf Valiant load balancing in hypercubes and other architectures:}
Leslie Valiant introduced oblivious routing in~\cite{valiant82}.
The VLB scheme for randomized routing in the hypercube was introduced, and shown to be optimal, by Valiant and Brebner~\cite{valiant81,valiant82}. While these works evaluate latency under queueing, they do not evaluate throughput.
Additionally, they use a direct-connect torus topology.
Our work can be interpreted as proving that VLB is
the optimal oblivious routing scheme to use in conjunction
with an optimally-designed reconfigurable network topology,
thus providing further theoretical justification for the
widespread usage of VLB in practice when oblivious
routing is applied on handcrafted network topologies.

A lower bound for {\em deterministic} oblivious routing in $d$-regular networks with $N$ nodes was proven in \cite{kaklam91}; the same paper shows this bound is tight for hypercube networks, in which $d = \log(N)$.


{\bf Load-Balanced Switches:}
The load-balanced switch architecture proposed by Chang \cite{lb-switch} uses static schedules and sends traffic obliviously via intermediate nodes.
While there are significant similarities between this architecture and ORNs, it differs in its use of specialized intermediate nodes (rather than sending traffic via multiple end-hosts), as well as its focus on monolithic switches.

{\bf Circuit-Switched Datacenter Network Architectures:}
c-Through \cite{c-through} and Traffic Matrix Scheduling \cite{microsecond-circuit-switching-dc}, as well as many other designs, propose a hybrid network in which a packet-switched backbone exists alongside a circuit-switched fabric.
However, with advances in circuit switches that have reduced reconfiguration times to nanosecond-scale, it is worth reconsidering whether a separate packet-switched backbone is truly necessary.

{\bf Oblivious Circuit-Switched Networks:}
Rotornet and Sirius \cite{rotornet, sirius} are two ORN concepts proposed for datacenter-wide networks that use optical circuit switches to build a reconfigurable network fabric.
Shoal \cite{shoal} is a similar ORN concept that uses electric circuit switches in a disaggregated rack environment.
Together, these works demonstrate that the ORN paradigm is feasible in practice.
These designs use similar schedules that prioritize achieving high throughput at the expense of poor latency for large $N$.
Our first ORN design, EBS, generalizes these existing designs to achieve many potential tradeoffs, ranging from the existing tradeoff to that achieved by an ORN version of hypercube routing.

Opera \cite{opera} evolves on the ORN concept by greatly lengthening each timeslot and creating an expander graph topology between nodes during each timeslot.
Opera uses a non-oblivious routing scheme in which latency-sensitive traffic is sent via multiple hops within a single expander graph topology, while throughput-sensitive traffic is held until the schedule advances to a topology in which it can be sent directly to the destination in one hop.
This design makes strong assumptions about the workload, including that bandwidth-sensitive traffic demand is near all-to-all, limiting its flexibility.

\section{Definitions} \label{sec:definitions}

This section presents definitions that formalize the notion
of an oblivious reconfigurable network (ORN). We assume a network of
$N$ nodes communicating in discrete, synchronous timeslots.
The nodes are joined by a communication medium that allows an
arbitrary pattern of unidirectional communication links to be
established in each timeslot, subject to a degree constraint
that each node participates as the sender in at most $d$
connections, and as the receiver in at most $d$ connections.
Throughout most of this paper we specialize to the case $d=1$;
see \Cref{sec:d-greater-than-one} below for a discussion of
why the general case reduces to this special case.

In systems that instantiate reconfigurable networking,
data is encapsulated in fixed-size units called {\em frames}
or {\em packets}. In this work we instead treat data as a
continuously-divisible commodity, and we allow sending
fractional quantities of flow along multiple paths from
the source to the destination. This abstraction
is standard in theoretical works on oblivious routing,
and it can be justified by interpreting a fractional flow as
a probability distribution over routing paths, with each
discrete frame being sent along one path sampled at random
from the distribution. Under this interpretation
flow values represent the expected number of frames
traversing a link.


\begin{dfn} \label{def:connection-schedule}
 A {\em connection schedule} $\bm{\pi}$ with size $N$ and period length $T$ is a sequence of permutations $\pi_0,\pi_1,\ldots,\pi_{T-1}$, each mapping $[N]$ to $[N]$. The interpretation of the relation $\pi_k(i) = j$ is that node $i$ is allowed to send one frame to node $j$ during any timeslot $t$ such that $t \equiv k \pmod{T}$.

 The {\em virtual topology} of the connection schedule $\bm{\pi}$ is a directed graph $\gpi$ with vertex set $[N] \times \mathbb{Z}$.
 The edge set of $\gpi$ consists of the union of $\evirt$ and $\ephys$.
 $\evirt$ is the set of {\em virtual edges}, which are of the form $(i,t)\to(i,t+1)$ and represent the frame waiting at node $i$ during the timeslot $t$.
 $\ephys$ is the set of {\em physical edges}, which are of the form $(i,t)\to(\pi_t(i),t+1)$ and represent the frame being transmitted from $i$ to $\pi_t(i)$ at timeslot $t$.
\end{dfn}

We interpret a path in $\gpi$ from $(a,t)$ to $(b,t^\prime)$ as a potential way to transmit a frame
from node $a$ to node $b$, beginning at timeslot $t$ and ending at some timeslot $t^\prime$.
For a node $a \in [N]$ let $\timeline{a}$ denote the set $\{a\} \times \mathbb{Z}$, consisting of all copies of $a$ in $\gpi$.
Let $\pths(a,b,t)$ denote the set of paths in $\gpi$ from the vertex $(a,t)$ to $\timeline{b}$. 
Finally, let $\pths = \bigcup_{a,b,t} \pths(a,b,t)$ denote the set of all paths in $\gpi$.

\begin{figure}
 \begin{subfigure}{0.38\columnwidth}%
   \center
   \begin{tabular}{cc|c|c|c|}
    &\multicolumn{1}{c}{}&\multicolumn{3}{c}{Timeslot}\\
    &\multicolumn{1}{c}{}&\multicolumn{1}{c}{0}&\multicolumn{1}{c}{1}&\multicolumn{1}{c}{2} \\
    \cline{3-5}
    \multirow{4}{*}{\rotatebox[origin=c]{90}{Node}}
    & A & B & C & D \\
    \cline{3-5}
    & B & C & D & A \\
    \cline{3-5}
    & C & D & A & B \\
    \cline{3-5}
    & D & A & B & C \\
    \cline{3-5}
   \end{tabular}
 \end{subfigure}
 \begin{subfigure}{0.4\columnwidth}%
   \includegraphics[width=\textwidth]{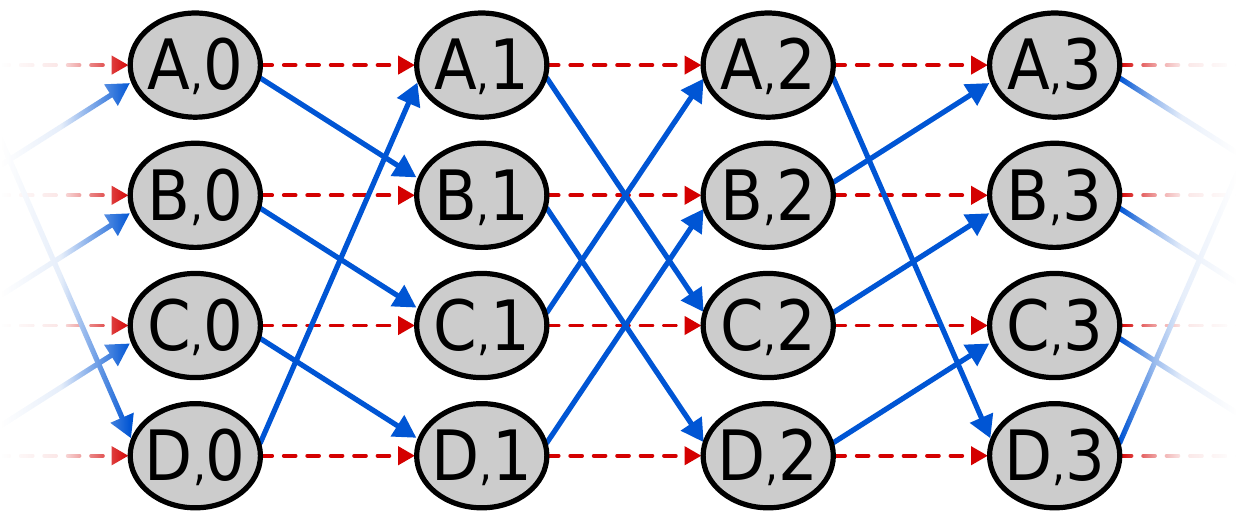}
 \end{subfigure}
 \caption{A connection schedule among four nodes, as well as part of its corresponding virtual topology. The full virtual topology represents a countably infinite number of timeslots.}
 \label{fig:simple-sched}
\end{figure}


\begin{dfn}
A {\em flow} is a function $f : \pths \to \rplus$.
For a given flow $f$, the amount of flow traversing an edge $e$ is defined as:
\[
  F(f,e) =
  \sum_{P \in \pths} f(P) \cdot \bm{1}_{e \in P}
\]
We say that $f$ is {\em feasible} if for every
 physical edge $e \in \ephys$, $F(f,e) \leq 1$.
\end{dfn}
\begin{dfn}
 The {\em latency} $L(P)$ of a path $P$ in $\gpi$ is equal to the number of edges it contains (both virtual and physical).
Note that traversing any edge in the virtual topology (either virtual or physical) is equivalent to advancing in time by the duration of one timeslot, so the number of edges in a path is proportional to the elapsed time.
For a nonzero flow $f$,
 the {\em maximum latency} is the maximum over all paths in the flow
 \[ L_{max}(f) = \max_{P\in\pths}\{ L(P) : f(P) > 0\} \]
\end{dfn}
We remark that our definitions of latency and of
the virtual topology $G_{\pi}$ incorporate the idealized assumption
of zero propagation delay. In other words, we assume that a frame
sent in one timeslot is received by the beginning of the following timeslot, and
that the number of edges of a path in the virtual topology accurately
reflects the length of the time interval
between when the frame originates and when
it reaches its destination.

\begin{dfn}\label{def:routing-scheme}
 An {\em oblivious routing scheme} $R$ is a function that associates to every $(a,b,t) \in [N]\times[N]\times\mathbb{Z}$ a flow $R_{a,b,t}$ such that:
  \begin{enumerate}
    \item $R_{a,b,t}$ is supported on paths from $(a,t)$ to $\timeline{b}$, meaning $\forall P \not\in \pths(a,b,t) \;\; R_{a,b,t}(P) = 0$.
    \item $R_{a,b,t}$ defines one unit of flow. In other words, $\sum_{P} R_{a,b,t}(P) = 1$.
    \item $R$ has period $T$. In other words, $R_{a,b,t+T}$ is equivalent to $R_{a,b,t}$ (except with all paths transposed by $T$ timeslots, as required to satisfy point 1).
  \end{enumerate}


\end{dfn}
\begin{dfn}
 A {\em demand matrix} is an $N\times N$ matrix which associates to each ordered pair $(a,b)$ an amount of flow to be sent from $a$ to $b$. A {\em demand function} $D$ is a function that associates to every $t \in \mathbb{Z}$ a demand matrix $D(t)$ representing the amount of flow $D(t,a,b)$ to originate between each source-destination pair $(a,b)$ at timeslot $t$.
The {\em throughput requested by demand function $D$} is the maximum,
 over all $t$, of the maximum row or column sum of $D(t)$.
\end{dfn}
\begin{dfn}
  For a given oblivious routing scheme $R$ and demand function $D$, the
  {\em induced flow} $f(R,D)$ is defined by:
  \[ f(R,D) = \sum_{(a,b,t) \in [N] \times [N] \times \mathbb{Z} } D(t,a,b) R_{a,b,t}. \]
\end{dfn}
\begin{dfn} \label{dfn:guarantees-throughput}
  An oblivious routing scheme is said to
  {\em guarantee throughput $r$}
  if the induced flow $f(R,D)$ is feasible whenever the demand function $D$ requests
  throughput at most $r$.
\end{dfn}
\Cref{dfn:guarantees-throughput} can be interpreted
as meaning that the network is able to simulate a ``big switch''
with $N$ input and output ports having line rate $r$: as long
as the amount of data originating at any node $a$ or destined
for any node $b$ does not exceed rate $r$ per timeslot, the
network is able to route all data to its destination without
violating capacity constraints.

In this work, we examine the tradeoffs between guaranteed throughput and maximum latency. Specifically, among ORNs of size $N$ that guarantee throughput $r$, what is the lowest possible maximum latency?  

%
%
%

\subsection{Allowing degree $d>1$ in a timeslot}
\label{sec:d-greater-than-one}

 Although our formalization of ORNs only describes networks in which nodes have a degree of $1$ in every timeslot, it can be generalized to networks that support a $d$-regular connectivity pattern in each timeslot.
 When $d>1$, we interpret a demand matrix $D$ which requests throughput $r$ as one in which the row and column sums of $D$ are bounded above by $dr$.

 The connectivity of $N\times \{t,t+1\}$ is $d$-regular bipartite. By K\H{o}nig's Theorem, this edge set can be decomposed into $d$ edge-disjoint perfect matchings, which we use to ``unroll'' into $d$ consecutive timeslots of a 1-regular ORN.
 Therefore, a $d$-regular ORN design which guarantees throughput $r$ with maximum latency $L$ unrolls into a 1-regular ORN design which guarantees throughput $r$ with maximum latency $dL$.

Under this framework, a lower bound $\lowerbd(r,N)$ for 1-regular ORN designs trivially implies the lower bound $\frac{1}{d} \lowerbd(r,N)$ for $d$-regular designs.
However, an upper bound for 1-regular designs does not necessarily imply a similar upper bound for $d$-regular designs, because the routing scheme could route paths
containing two or more physical edges in timeslots belonging to the same ``unrolled'' segment of the 1-regular virtual topology. This would correspond to traversing two or more edges at once in the $d$-regular topology.
We show in \Cref{sec:upper-bound} that such a problem will never occur due to our construction. Specifically, we show that our construction can be modified to never allow flow to be routed along two edges within any block of $d$ consecutive time slots, provided $d \le N^{1/(h+1)}$. This modification will add a factor of at most 2 to the maximum latency. Then, by inverting the unrolling process, we will obtain a $d$-regular ORN design with maximum latency $L = O(\frac1d \lowerbd(r,N))$. This confirms that the tight bound on maximum latency for $d$-regular ORN designs is $\Theta(\frac1d \lowerbd(r,N))$
whenever $d \le N^{1/(h+1)}$ and justifies our focus on the case $d=1$ throughout the remainder this paper.

%
%


\section{Lower Bound} \label{sec:lower-bound}

In this section we prove the lower-bound half of \Cref{thm:lb-informal}, which says that when $\frac{1}{2r} = h+1-\eps$ with $h \in \mathbb{N}$ and $0 < \eps \le 1$, any $d$-regular, $N$-node ORN design that guarantees throughput $r$ must have maximum latency $\Omega(\frac{h}{d} [ N^{1/(h+1)} + (\eps N)^{1/h} ])$.
As noted in \Cref{sec:d-greater-than-one}, the general case of this lower bound reduces to the case $d=1$, and we will assume $d=1$ throughout the remainder of this section.

Because the full proof is somewhat long, we begin by sketching some of the main ideas in the proof, beginning with a much simpler argument leading to a lower bound of the form $\Omega(\frac1r N^{r})$ when $1/r$ is an integer.
This simple lower bound applies not only to oblivious routing schemes, but to {\em any} feasible flow $f$ that solves the uniform multicommodity flow problem given by the demand function $D(t,a,b) = \frac{r}{N-1}$ for all $t \in [T]$ and $b \neq a$.
The lower bound follows by combining a few key observations.

\begin{enumerate}
	\item Define the cost of a path to be the number of physical edges it contains.
		Since every source sends out $r$ units of flow at all times, the flow $f$ sends out $r N T$ units of flow per $T$-step period, in a network whose physical edges have only $N T$ units of capacity per $T$-step period.
		Consequently the average cost of flow paths in $f$ must be at most $\frac1r$.
	\item For any source node $(a,t)$ in the virtual topology, the number of distinct destinations $\timeline{b}$ that can be reached via a path with maximum latency $L$ and cost $p$ is bounded above by $\binom{L}{p}$. 
	\item If $L \le \frac{1}{2er} N^{r}$, we have $\binom{L}{1/r} \le N/4$ and $\sum_{p=1}^{1/r} \binom{L}{p} \le N/2$, so the majority of source-destination pairs cannot be joined by a path with latency $L$ and cost less than $\frac1r + 1$.
		In fact, even if we connect every source and destination with a minimum-cost path (subject to latency bound $L$), one can show that the average cost of paths will exceed $\frac{1}{r}$.
	\item Since a feasible flow must have average path cost
		at most $\frac{1}{r}$, we can conclude that a feasible flow
		does not exist when  $L \le \frac{1}{2er} N^{r}$.
\end{enumerate}

When $1/r$ is an integer, this lower bound of $L_{max} \geq \frac{1}{2er} N^{r}$ for feasible uniform multicommodity flows turns out to be tight up to a constant factor. 
However for oblivious routing schemes, \Cref{thm:lb-informal} shows that maximum latency is bounded below by a function in which the exponent of $N$ is roughly twice as large. 
Stated differently, for a given maximum latency bound, the optimal throughput guarantee for oblivious routing is only half as large as the throughput of an optimal uniform multicommodity flow.

The factor-two difference in throughput between oblivious routing and optimal uniformly multicommodity flow solutions aligns with the intuition that oblivious routing schemes must use indirect paths (as in Valiant load balancing) if they are to guarantee throughput $r$, whereas uniform multicommodity flow solutions (in a well-designed virtual topology) can afford to satisfy all demands using shortest-path routing.
The proof of the lower bound for oblivious routing needs to substantiate this intuition.

To do so, we formulate oblivious routing as a linear program and interpret the dual variables as specifying a more refined way to measure the cost of paths.
Rather than defining the cost of a path to be its number of physical edges, the duality-based proof amounts to an accounting system in which the cost of using an edge depends on the endpoints of the path in which the edge is being used.
For a parameter $\theta$ which we will set to $h+1$ (unless $\eps$ is very small, in which case we'll set $\theta=h+2$), the dual accounting system assesses the cost of an edge to be 1 if its distance from the source is less than $\theta$, plus 1 if its distance from the destination is less than $\theta$.
Thus, the cost of an edge is doubled when it is close to both the source and the destination.
The doubling has the effect of equalizing the costs of direct and indirect paths: when the distance between a source and destination is at least $\theta$, there is no difference in cost between a shortest path and one that combines two semi-paths each composed of $\theta$ physical edges.

Viewed in this way, it is intuitive that the proof manages to show that VLB routing schemes, which construct routing paths by concatenating random semi-paths with the appropriate number of physical edges, correspond to optimal solutions of the oblivious routing LP.
The difficulty in the proof lies in showing that the constructed dual solution is feasible; for this, we make use of a version of the same counting argument sketched above, that bounds the number of distinct destinations reachable from a given source under constraints on the maximum latency and the maximum number of physical edges used.

\subsection{Lower Bound Theorem Proof}

Before presenting the proof of \Cref{thm:counting-lb}, we formalize the counting argument we reasoned about in our proof sketch.

\begin{lem}\label{lem:counting-lem}
\textbf{(Counting Lemma)} If in an ORN topology, some node $a$ can reach $k$ other nodes in at most $L$ timeslots using at most $h$ physical hops per path for some integer $h$, then $k\leq 2 {L \choose h}$, assuming $h\leq \frac{1}{3}L$.
\end{lem}
\begin{proof}
If node $a$ can reach $k$ other nodes in $\leq L$ timeslots using exactly $h$ physical hops per path, then $k\leq {L \choose h}$. Additionally, the function
${L \choose h}$ grows at least exponentially in base 2
--- that is, ${L \choose h} \geq 2 {L \choose h-1}$ ---
up until $h = \frac{1}{3}L$. Therefore, the number of such $k$ is
at most $\sum_{i=1}^{h} {L \choose i} \leq 2{L \choose h}$.
\end{proof}

\begin{thm}\label{thm:counting-lb}
Given an ORN design $\rout$ which guarantees throughput $r$, the maximum latency suffered by any routing path $P$ with $R_{a,b,t}(P) > 0$ over all $a,b,t$ is bounded by the following equation
\begin{equation}
L_{max} \geq \Omega\left(h \left[(\eps N)^{1/h} + N^{1/(h+1)}\right]\right)
\end{equation}
where $h = \floor{\frac{1}{2r}}$ and $\eps \in(0,1]$ is set to equal $h+1 - \frac{1}{2r}.$
In other words, $(h,\eps)$ is the unique solution in
$\mathbb{N} \times (0,1]$ to the equation $\frac{1}{2r} = h+1-\eps$.
\end{thm}


\begin{proof}
Consider the linear program below which maximizes throughput given a maximum latency constraint, $L$, where we let $\pths_L(a,b,t)$ be the set of paths from $(a,t)\rightarrow \timeline{b}$ with latency at most $L$.
\begin{center}\fbox{\begin{minipage}{0.92\textwidth}
	\textbf{LP}
	\[\begin{lparray}
    \mbox{maximize} & r \\
    \mbox{subject to} & \sum_{P\in\pths_L(a,b,t)} R_{a,b,t}(P) = r \hfill \forall a,b\in[N],\text{ } t\in[T] \\
    & \sum_{a\in[N]}\sum_{t=0}^{T-1}\sum_{P\in\pths_L(a,\sigma(a),t) : e\in P} R_{a,\sigma(a),t}(P) \leq 1  \hfill \quad \qquad\forall \sigma \in S_N,\hfill e\in\ephys \\
    & R_{a,b,t}(P) \geq 0 \hfill \forall a,b\in[N], \text{ }t\in[T], \text{ } P\in \pths_L(a,b,t)
  \end{lparray}
\] \end{minipage} }
\end{center}

The second set of constraints, in which the parameter $\sigma$ ranges
over the set $S_N$ of all permutations of $[N]$, can be reformulated as the following set of nonlinear constraints in which the maximum is again taken over all permutations $\sigma$:
\[ \max_{\sigma}\left\{r\sum_{a\in[N]}\sum_{t=0}^{T-1}\sum_{P\in \pths_L(a,\sigma(a),t) : e\in P} R_{a,\sigma(a),t}(P)\right\} \leq 1 \text{ }\text{ } \forall e\in\ephys\]
Note that given an edge $e$, this maximization over permutations $\sigma$ corresponds to maximizing over perfect bipartite matchings with edge weights defined by $w_{a,b,e} = \sum_{t=0}^{T-1}\sum_{P\in\pths_L(a,b,t) : e\in P} R_{a,b,t}(P)$. This prompts the following matching LP and its dual.

\begin{center}
\fbox{\begin{minipage}{0.47\textwidth}
	\textbf{Matching LP}
	\[ \begin{lparray}
    \mbox{maximize} & \sum_{a, b} u_{a,b,e} w_{a,b,e} \\
    \mbox{subject to} & \sum_{b\in[N]} u_{a,b,e} \leq 1 \hfill\forall a\in[N] \\
      & \sum_{a\in[N]} u_{a,b,e} \leq 1 \hfill\forall b\in[N] \\
      & u_{a,b,e} \geq 0 \hfill \;\; \forall a,b\in[N], e\in\ephys \\
  \end{lparray}
\]
\end{minipage} }
\fbox{\begin{minipage}{0.47\textwidth}
	\textbf{Matching Dual}
	\[
  \begin{lparray}
    \mbox{minimize} & \sum_{a\in[N]} \xi_{a,e} + \sum_{b\in[N]} \eta_{b,e} \\
    \mbox{subject to} & \xi_{a,e} + \eta_{b,e} \geq w_{a,b,e} \hfill \;\; \forall a,b\in[N] \\
      & \xi_{a,e} \geq 0 \hfill\forall a\in[N],e\in\ephys \\
      & \eta_{b,e} \geq 0 \hfill\forall b\in[N],e\in\ephys \\
  \end{lparray}
\]
\end{minipage}}
\end{center}

We then substitute finding a feasible
matching dual solution into the original LP,
replace the expression $w_{a,b,e}$
with its definition $\sum_{t=0}^{T-1}\sum_{P\in\pths_L(a,b,t) : e\in P} R_{a,b,t}(P)$,
and take the dual again.

\begin{center}\fbox{\begin{minipage}{0.90\textwidth}
	\textbf{LP}
	\[\begin{lparray}
    \mbox{maximize} & r \\
    \mbox{subject to} & \sum_{P\in\pths_L(a,b,t)} R_{a,b,t}(P) = r \hfill \forall a,b\in[N], \text{ }t\in[T] \\
		& \xi_{a,e} + \eta_{b,e} \geq \sum_{t=0}^{T-1}\sum_{P\in\pths_L(a,b,t) : e\in P} R_{a,b,t}(P) \hfill  \qquad \forall a,b\in[N],e \in \ephys \\
		& \sum_{a \in [N]} \xi_{a,e} + \sum_{b \in [N]} \eta_{b,e} \le 1 \hfill \forall e \in \ephys \\
    & \xi_{a,e} \geq 0 \hfill \forall a\in[N],e\in\ephys \\
    & \eta_{b,e} \geq 0 \hfill \forall b\in[N],e\in\ephys \\
    & R_{a,b,t}(P) \geq 0 \hfill \forall a,b\in[N], \text{ }t\in[T],\text{ }P\in \pths_L(a,b,t)
  \end{lparray}
\] \end{minipage} }
\end{center}

\begin{center}\fbox{\begin{minipage}{0.90\textwidth}
	\textbf{Dual}
	\[\begin{lparray}
    \mbox{minimize} & \sum_{e} z_e \\
    \mbox{subject to} & \sum_{a,b,t} x_{a,b,t} \geq 1 \\
      & z_e \geq \sum_{b} y_{a,b,e} \hfill \text{ }\forall a\in[N],e\in\ephys \\
      & z_e \geq \sum_{a} y_{a,b,e} \hfill \text{ }\forall b\in[N],e\in\ephys \\
      & \sum_{e\in P} y_{a,b,e} \geq x_{a,b,t} \hfill \quad \qquad \forall a,b\in[N], \text{ }t\in[T],\text{ }P\in \pths_L(a,b,t) \\
      & y_{a,b,e},z_e \geq 0 \hfill \text{ }\forall a,b\in[N], \text{ }e\in\ephys
  \end{lparray}
\]\end{minipage} }
\end{center}

The variables $y_{a,b,e}$ can be interpreted as either edge costs we assign dependent on source-destination pairs $(a,b)$, or demand functions designed to overload a particular edge $e$. We will use both interpretations, depending on if we are comparing $y_{a,b,e}$ variables to either $x_{a,b,t}$ or $z_e$ variables respectively. According to the fourth dual constraint, the variables $x_{a,b,t}$ can be interpreted as encoding the minimum cost of a path from $(a,t)$
to $\timeline{b}$ subject to latency bound $L$.
According to the second and third dual constraints, the variables $z_e$ can be interpreted as bounding the throughput requested by the demand function $D(t,a,b) = y_{a,b,e}$.
We will next define the cost inflation scheme we use to set our dual variables. \\

\textbf{Cost inflation scheme } For a given node $a \in [N]$ and cutoff $\theta\in\mathbb{Z}_{>0}$, we will classify edges $e$ according to whether they are reachable within $\theta$ physical hops of $a$, counting edge $e$ as one of the hops. (In other words, one could start at node $a$ and cross edge $e$ using $\theta$ or fewer physical hops.) We define this value $m^{+}_{\theta}(e,a)$ as follows.
\begin{align*}
m^{+}_{\theta}(e,a) & = \begin{cases}
  1 & \mbox{if } e \text{ can be reached from } a \text{ using at most } \theta \text{ physical hops (including } e \text{)} \\
  0 & \mbox{if } \text{ otherwise}
\end{cases}
\end{align*}
We define a similar value for edges which can reach node $b$.
\begin{align*}
m^{-}_{\theta}(e,b) & = \begin{cases}
  1 & \mbox{if } b \text{ can be reached from } e \text{ using at most } \theta \text{ physical hops (including } e \text{)} \\
  0 & \mbox{if } \text{ otherwise}
\end{cases}
\end{align*}
To understand how these values are set, consider some path $P$ from $(a,t)\rightarrow \timeline{b}$. If we consider the $m^{+}_{\theta},m^{-}_{\theta}$ weights on the edges of $P$, then the first $\theta$ physical hop edges of $P$ have weight $m^{+}_{\theta}(e,a) = 1$ and the last $\theta$ physical hop edges of $P$ have weight $m^{-}_{\theta}(e,b) = 1$. It may be the case that some edges have both $m^{+}_{\theta}(e,a) = m^{-}_{\theta}(e,b) = 1$, if $P$ uses fewer than $2\theta$ physical hops. And if $P$ uses $\theta$ or fewer physical hops, then every physical hop edge along $P$ has weight $m^{+}_{\theta}(e,a) = m^{-}_{\theta}(e,b) = 1$. All other weights may be $0$ or $1$ depending on whether those edges are otherwise reachable from $a$ or can otherwise reach $b$.

We start by setting $\hat{y}_{a,b,e} = m^{+}_{\theta}(e,a) + m^{-}_{\theta}(e,b)$. Also set $\hat{x}_{a,b,t} = \min_{P\in \pths_L(a,b,t)} \{ \sum_{e\in P} \hat{y}_{a,b,e} \}$. Note that by definition, $\hat{x}$ and $\hat{y}$ variables satisfy the last dual constraint. We will next find a lower bound $w\leq\sum_{a,b,t} \hat{x}_{a,b,t}$ and use that to normalize the $\hat{x},\hat{y}$ variables to satisfy  the first dual constraint.

Note that $\sum_{e\in P} \hat{y}_{a,b,e} \geq \min \{2\theta, 2|P\cap\ephys|\}$. Then we can bound the sum of $\hat{x}$ variables by

\[\sum_{a,b,t} \hat{x}_{a,b,t} \geq \sum_{a,t}\sum_{b\neq a} \min_{P\in \pths_L(a,b,t)} \{2\theta, 2|P\cap\ephys|\}\]

Note that $\hat{x}_{a,b,t}<2\theta$ only when there exists some path from $(a,t)$ to $\timeline{b}$ which uses less than $\theta$ physical edges. We can then use the Counting Lemma to produce an upper bound on the number of $b\neq a$ which have such paths: this is at most $2{L\choose \theta-1}$.

So, assuming that $2 {L\choose \theta-1}\leq N$ and that $\theta-1 \leq L/3$, we have

\[\sum_{a,t}\sum_{b\neq a} \hat{x}_{a,b,t} \geq NT \left( 2\theta \left( N-2 {L\choose \theta-1} \right) + {L\choose\theta-1} \right)\]

Set
\[
	w = NT \left( 2\theta \left( N-2 {L\choose \theta-1} \right) + {L\choose\theta-1} \right),
\]
and then set $y_{a,b,e} = \frac{1}{w}\hat{y}_{a,b,e}$ and $x_{a,b,t} = \frac{1}{w} \hat{x}_{a,b,t}$.

Next, we set $z_e = \max_{a,b} \{ \sum_{a}y_{a,b,e},\sum_{b}y_{a,b,e} \}$. By construction, the values of $x_{a,b,t}, y_{a,b,e}, z_e$ that we have defined satisfy the dual constraints. Then to bound throughput from above, we upper bound the sums $\sum_{a}y_{a,b,e}$ and $\sum_{b}y_{a,b,e}$, thus upper bounding the sum of $z_e$'s.

\begin{align*}
    \sum_{a}y_{a,b,e} & = \frac{1}{w} \sum_{a} \left(m^{+}_{\theta}(e,a) + m^{-}_{\theta}(e,b)\right)
      \leq \frac{1}{w}\left( \sum_{a} m^{+}_{\theta}(e,a) + N-1\right)
      \leq \frac{1}{w}\left(2{L\choose \theta-1} +N-1\right)
\end{align*}
where the last step is an application of the
Counting Lemma. Similarly,
\begin{align*}
    \sum_{b}y_{a,b,e} & = \frac{1}{w} \sum_{b} \left(m^{+}_{\theta}(e,a) + m^{-}_{\theta}(e,b)\right)
     \leq \frac{1}{w}\left( N-1 + \sum_{b} m^{-}_{\theta}(e,b) \right)
     \leq \frac{1}{w}\left(N-1 + 2{L\choose \theta-1} \right)
\end{align*}
Recalling that $z_e = \max_{a,b} \{ \sum_{a}y_{a,b,e},\sum_{b}y_{a,b,e} \}$, we deduce that
\[
	z_e \leq \frac{1}{w}\left(N-1 + 2{L\choose \theta-1} \right) .
\]
Using this upper bound on $z_e$, we find that the optimal value of the dual objective --- hence also the optimal value of the primal, i.e.~the maximum throughput of oblivious routing schemes --- is bounded by
\begin{align*}
    r & \leq \sum_{e} z_e \leq \frac{NT}{w}\left(N-1 + 2{L\choose \theta-1}\right) \\
    & = \frac{N-1 + 2{L\choose \theta-1}}{2\theta N - 4\theta {L\choose \theta-1} + 2{L\choose \theta-1}} \\
    & \leq \frac{N-1 + 2{L\choose \theta-1}}{2\theta N - 4\theta {L\choose \theta-1}} \\
    & = \frac{N-1 + \frac{2(L!)}{(\theta-1)! (L-\theta+1)!}}{2\theta N - 4\theta\frac{L!}{(\theta-1)! (L-\theta+1)!}} \\
    & = \frac{(N-1)(\theta-1)!(L-\theta+1)! + 2(L!)}{2\theta ( N(\theta-1)!(L-\theta+1)! -2(L!) )} \\
    & = \frac{1}{2\theta} + \frac{4(L!)}{2\theta(L-\theta+1)! \left( N(\theta-1)!-2\frac{L!}{(L-\theta+1)!}\right)} \\
    & \leq \frac{1}{2\theta} + \frac{4L^{\theta-1}}{2\theta(N(\theta-1)! - 2L^{\theta-1})}
\end{align*}
using the fact that $\frac{a!}{(a-b)!}\leq a^b$. At this point, we
can rearrange the inequality to isolate $L$.
\begin{align*}
    r - \frac{1}{2\theta} & \leq \frac{4L^{\theta-1}}{2\theta(N(\theta-1)! - 2L^{\theta-1})} \\
    \left(r - \frac{1}{2\theta}\right) \left(2\theta N(\theta-1)!\right) - \left(r - \frac{1}{2\theta}\right) 4\theta L^{\theta-1} & \leq 4L^{\theta-1} \\
    \left(r - \frac{1}{2\theta}\right)2\theta N(\theta-1)! & \leq L^{\theta-1}\left(4 + \left(r - \frac{1}{2\theta}\right)4\theta\right) \\
    \frac{(r - \frac{1}{2\theta})2\theta N(\theta-1)!}{4 + (r - \frac{1}{2\theta})4\theta} & \leq L^{\theta-1} \\
    \left(\frac{(r - \frac{1}{2\theta})2\theta N(\theta-1)!}{4 + (r - \frac{1}{2\theta})4\theta}\right)^{\frac{1}{\theta-1}} & \leq L
\end{align*}
Now that we have a closed form, we simplify. We use Stirling's approximation, in the form $(k!)^{\frac{1}{k}} \geq \frac{k}{e}\sqrt{2\pi k}^{\frac{1}{k}}$.

\begin{align*}
    L & \geq \left(\frac{(r - \frac{1}{2\theta})2\theta N(\theta-1)!}{4 + (r - \frac{1}{2\theta})4\theta}\right)^{\frac{1}{\theta-1}} \\
    & = N^{\frac{1}{\theta-1}} (\theta-1)!^{\frac{1}{\theta-1}} \left(\frac{(r - \frac{1}{2\theta})2\theta}{4 + (r - \frac{1}{2\theta})4\theta}\right)^{\frac{1}{\theta-1}} \\
    & \geq \frac{\theta-1}{e} N^{\frac{1}{\theta-1}} \left(\frac{(r - \frac{1}{2\theta})2\theta \sqrt{2\pi (\theta-1)}}{4 + (r - \frac{1}{2\theta})4\theta}\right)^{\frac{1}{\theta-1}}
   	\geq \frac{\theta-1}{e} N^{\frac{1}{\theta-1}} \left(\frac{(r - \frac{1}{2\theta})\theta \sqrt{\frac{\pi(\theta-1)}{2}}}{\theta r + \frac{1}{2}}\right)^{\frac{1}{\theta-1}}
\end{align*}
To set the parameter $\theta$, first note that the above bound is positive when $r > \frac{1}{2\theta}$. Additionally, we would like to set $\theta$ as large as possible, and $\theta$ must be an integer value (otherwise the Counting Lemma doesn't make sense). Taking this into account, we set $\theta = \floor{\frac{1}{2r}}+1$, the nearest integer for which $(r-\frac{1}{2\theta})$ produces a positive value.

To simplify our lower bound further, let $h = \floor{\frac{1}{2r}}$ and $\eps = h+1 - \frac{1}{2r}$. These can be interpreted in the following way: $h$ represents the largest number of physical hops we take per path (approximately), and $\eps$ is directly related to how many pairs take paths using $h$ physical hops instead of paths using fewer than $h$ physical hops. Note that $\eps\in(0,1]$. This gives the restated bound below.

\begin{align}
	L & \geq \frac{h}{e} N^{1/h} \left( \frac{(r-\frac{1}{2(h+1)}) (h+1) \sqrt{\frac{\pi h}{2}}} {(h+1)r + \frac{1}{2}} \right)^{1/h} \nonumber\\
	& = \frac{h}{e} N^{1/h} \left( \frac{\left(\frac{\eps}{2(h+1)(h+1-\eps)}\right) (h+1) \sqrt{\frac{\pi h}{2}}}{ 1 + \frac{\eps}{2(h+1-\eps)} } \right) ^{1/h} \nonumber\\
	& = \frac{h}{e} N^{1/h} \left( \frac{\eps \sqrt{\frac{\pi h}{2}}}{2(h+1-\eps) + \eps)} \right)^{1/h} \nonumber\\
        & \geq \frac{h}{e} (\eps N)^{1/h} \left( \frac{\sqrt{\frac{\pi h}{2}}}{4h} \right) ^{1/h} \label{eq:i-dont-know}\\
	& = \frac{h}{e} (\eps N)^{1/h} \cdot \Omega(1) = \Omega\left(h (\eps N)^{1/h}\right) \nonumber
\end{align}

As $\eps\rightarrow 0$, this bound goes toward 0, making it meaningless for extremely small values of $\eps$. However, for such values of $\eps$, we simply set $\theta = h+2$ instead, which gives the following

\begin{equation*}
	L_{max} \geq \Omega\left((h+1) N^{1/(h+1)}\right)
\end{equation*}
To combine the two ways in which we set $\theta$, we take the average of the two bounds. This gives the bound from our theorem statement,

\begin{equation*}
	L_{max} \geq \Omega\left(h \left[(\eps N)^{1/h} + N^{1/(h+1)}\right]\right)
	= \Omega \left( \lowerbd(r,N) \right).
\end{equation*}

\end{proof}




\section{Upper Bound} \label{sec:upper-bound}

To prove an upper bound on the latency achievable while guaranteeing a given throughput, we define an infinite family of ORN designs which we refer to as the Elementary Basis Scheme (EBS). The upper bound given by EBS is within a constant factor of the previously described lower bound for most values of $r$. To tightly bound the remaining values of $r$, we describe a second infinite family of ORN designs which we refer to as the Vandermonde Bases Scheme (VBS). Combined, EBS and VBS give a tight upper bound on maximum latency for all constant $r$. We address the upper bound for $d$-regular networks with $d>1$ by modifying EBS and VBS in \Cref{sec:d-regular-upper-bound}.

\subsection{Elementary Basis Scheme}
\paragraph{Connection Schedule:} \label{sec:shale-sched}
In EBS's connection schedule, each node participates in a series of sub-schedules called round robins. Consider a cyclic group $H = \mathbb{Z}/(n)$
acting freely on a set $S$ of $n$ nodes, where we denote the action of $t \in H$ on $i \in S$ by $i+t$. A round robin for $S$ is a schedule of $n-1$ timeslots in which each element of $S$ has a chance to send directly to each other element exactly once; during timeslot $t \in [n-1]$ node $i$ may send to $i+t$.
The number of round-robins in which each EBS node participates is controlled by a tuning parameter $\ell$ which we refer to as the \textit{order}. Similar to the previous section, $\ell$ will be half of the the maximum number of physical hops in an EBS path.

Let $n = N^{1/\ell}$, so that the node set $[N]$
is in one-to-one correspondence with the elements of the group $H^{\ell}$.
Each node $a \in [N]$ is assigned a unique set of $\ell$ coordinates $(a_0,a_1,...,a_{\ell-1}) \in H^{\ell}$
and participates in $\ell$ round robins, each containing the $n$ nodes that match in all but one of the $\ell$ coordinates. We refer to these round robins as \textit{phases} of the EBS schedule.
One full iteration of the EBS schedule, or \textit{epoch}, contains $\ell$ phases.
Because each phase is a round robin among $n$ nodes, each phase takes $n-1$ timeslots, resulting in an overall epoch length of $T=\ell(n-1) = \ell (N^{1/\ell}-1)$.

We now describe the EBS schedule formally. We express each node $\bm{i}$ as the $\ell$-tuple $(i_0, i_1, \ldots, i_{\ell-1}) \in (\mathbb{Z}/n)^{\ell}$. Similarly, we identify each permutation $\pi_k$ of the connection schedule using a scale factor $s$, $1 \leq s < n$, and a phase number $p$, $0 \leq p < \ell$, such that $k = (n-1) p + s - 1$. Let $\mathbf{e}_p$ denote the standard basis vector whose $p^{\mathrm{th}}$ coordinate is 1 and all other coordinates are 0. The connection schedule is then $\pi_{(n-1) p + s - 1}(\bm{i}) = \bm{i} + s \mathbf{e}_p = \bm{j}$. Since $\mathbf{e}$ is the standard basis, $j_x = i_x$ for $x \neq p$, and $j_p = i_p + s \pmod{n}$.



The EBS schedule can be seen as simulating a flattened butterfly graph between nodes \cite{flattened-butterfly}.
This schedule generalizes existing ORN designs which have thus far all been based on the same schedule: a single round robin among all nodes, simulating an all-to-all graph. When $\ell = 1$, the EBS schedule reduces to this existing schedule.
On the other hand, when $\ell = \log_2(N)$, the EBS schedule simulates a direct-connect hypercube topology. By varying $\ell$, in addition to achieving these two known points, the EBS family includes schedules which achieve intermediate throughput and latency tradeoff points.

\subsubsection{Oblivious Routing Scheme} \label{sec:shale-ors}

\begin{figure}
 \centering
  \begin{subfigure}{0.38\textwidth}
   \centering
   \begin{tabular}{cc|c|c||c|c|}
    &\multicolumn{1}{c}{}&\multicolumn{4}{c}{Timeslot}\\
    &\multicolumn{1}{c}{}&\multicolumn{1}{c}{0}&\multicolumn{1}{c}{1}&\multicolumn{1}{c}{2}&\multicolumn{1}{c}{3}\\
    \cline{3-6}
    \multirow{6}{*}{\rotatebox[origin=c]{90}{Node}}
    &A,A & B,A & C,A & A,B & A,C \\
    \cline{3-6}
    &B,A & C,A & A,A & B,B & B,C \\
    \cline{3-6}
    &C,A & A,A & B,A & C,B & C,C \\
    \cline{3-6}
    &$\cdot$ & $\cdot$ & $\cdot$ & $\cdot$ & $\cdot$ \\
    \cline{3-6}
    &$\cdot$ & $\cdot$ & $\cdot$ & $\cdot$ & $\cdot$ \\
    \cline{3-6}
    &B,C & C,C & A,C & B,A & B,B \\
    \cline{3-6}
    &C,C & A,C & B,C & C,A & C,B \\
    \cline{3-6}
   \end{tabular}
  \end{subfigure}
  \begin{subfigure}{0.6\textwidth}
   \includegraphics[width=\textwidth]{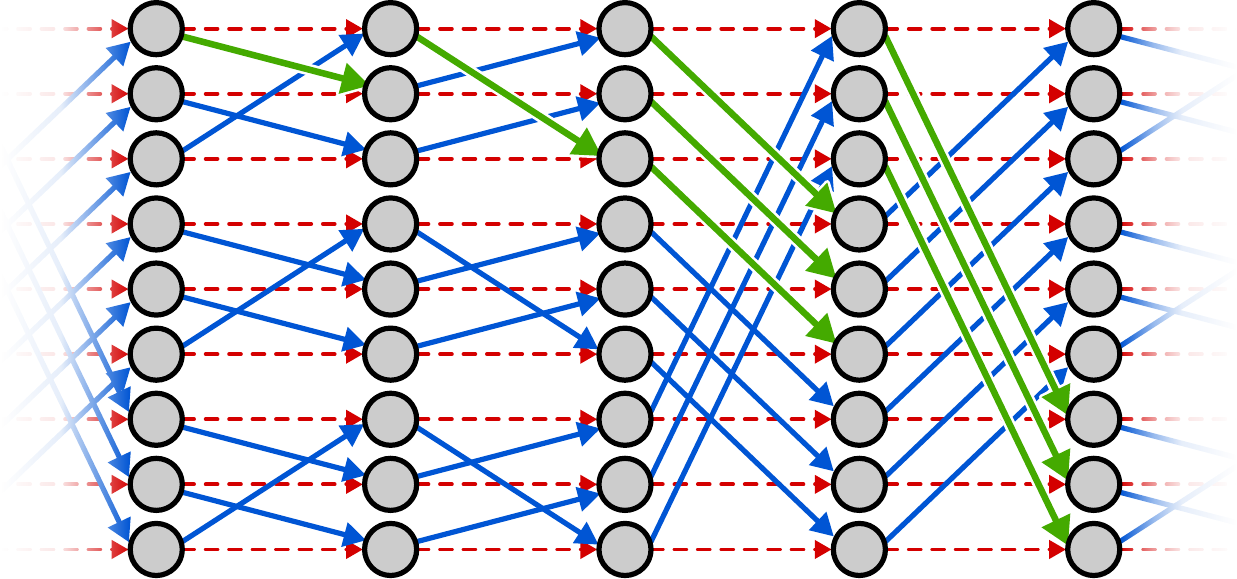}
  \end{subfigure}
   \caption{Connection schedule for 9 nodes in $\ell=2$ EBS, as well as part of the corresponding virtual topology. Physical edges used on semi-paths from ((A,A),0) to other nodes are highlighted in green. This schedule can be seen as a generalization of the one presented in \Cref{fig:simple-sched}.}
   \label{fig:shale-sched}
\end{figure}

The EBS oblivious routing scheme is based around Valiant load balancing (VLB)~\cite{vlb}. VLB operates in two stages: first, traffic is routed from the source to a random intermediate node in the network. Then, traffic is routed from the intermediate node to its final destination. This two-stage design ensures that traffic is uniformly distributed throughout the network regardless of demand. We refer to the path taken during an individual stage as a \textit{semi-path}, and we use the same algorithm to generate semi-paths in either stage.


To create a semi-path between a node $(a,t)$ and $\timeline{b}$, the following greedy algorithm is used starting at $(a,t)$: for the current node in the virtual topology, if the outgoing physical edge leads to a node with a decreased Hamming distance to $b$ (i.e. it matches $b$ in the modified coordinate), traverse the physical edge. Otherwise, traverse the virtual edge. This algorithm terminates when it reaches a node in $\timeline{b}$. Note that because there are $\ell$ coordinates, the largest Hamming distance possible is $\ell$, and the longest semi-paths use $\ell$ physical links.

In order to construct a full path from $(a,t)$ to $\timeline{b}$, first select an intermediate node $c$ in the system uniformly at random. Then, traverse the semi-path from $(a,t)$ to $\timeline{c}$. Let $t^\prime$ be the timeslot at which we reach $\timeline{c}$. If $t^\prime < t + T$, traverse virtual edges until node $(c,t + T)$ is reached. Finally, traverse the semi-path from $(c,t+T)$ to $\timeline{b}$.

The EBS oblivious routing scheme is formed as follows: for $R_{a,b,t}$, for all intermediate nodes $c$, construct the path from $(a,t)$ to $\timeline{b}$ via $c$ as described above, and assign it the value $\frac{1}{N}$. Assign all other paths the value $0$. Because there are $N$ possible intermediate nodes, each of which is used to define one path from $(a,t)$ to $\timeline{b}$, this routing scheme defines one unit of flow.

\subsection{Latency-Throughput Tradeoff of EBS}

\begin{prop}\label{thm:shale-informal}
  For each $r \le \frac12$ such that $\ell = \frac{1}{2r}$ is an integer,
  and each $N > 1$ such that $N^{1/\ell}$ is an integer, the EBS
  design of order $\ell$ on $N$ nodes guarantees throughput $r$
  and has maximum latency $\frac{1}{r} \left(N^{2r} - 1 \right)$.
\end{prop}

The proof of \Cref{thm:shale-informal} is contained in the following two subsections, which address the latency and throughput guarantees respectively.

\subsubsection{Latency} \label{sec:shale-latency}

Recall that $\ell = \frac{1}{2r}$ and that $n = N^{1/\ell} = N^{2r}$,
so the latency bound in \Cref{thm:shale-informal} can be written as
$2 \ell (n-1)$. Since the epoch length is $T = \ell(n-1)$, the latency
bound asserts that every EBS routing path completes within a time interval
no greater than the length of two epochs. An EBS path is composed of two
semi-paths, so we only need to show that each semi-path completes
within the length of a single epoch.

Let $(a,t)$ denote the first node of the
semi-path. If $t$ occurs at the start of a phase,
then after $p$ phases have completed the
Hamming distance to the semi-path's destination address
must be less than or equal to $t-p$; consequently
the semi-path completes after at most $\ell$ phases,
as claimed. If $t$ occurs in the middle of a
phase using basis vector $\mathbf{e}_p$,
let $s$ denote the number of timeslots that
have already elapsed in that phase.
Either the semi-path is able to match the
$p^{\mathrm{th}}$
destination coordinate before the phase ends,
or the coordinate can be matched during the first
$s$ timeslots of the next phase that uses basis
vector $\mathbf{e}_p$. In either case, the
$p^{\mathrm{th}}$ destination coordinate will
be matched no later than timeslot $t+T$, and
all other destination coordinates will be matched
during the intervening phases.

\subsubsection{Throughput} \label{sec:shale-throughput}

\begin{lem}\label{thm:shale-throughput}
Let $R$ be the EBS routing scheme for a given $N$ and $\ell$. For all demand functions $D$
 requesting throughput at most $\frac{1}{2\ell}$, the flow $f(R,D)$ is feasible.
\end{lem}

\begin{proof}
 Consider an arbitrary demand function $D$
 requesting throughput $r=\frac{1}{2\ell}$, and consider an arbitrary physical edge $e \in \ephys$ from $(i,t_e)$ to $(j,t_e+1)$, where $t_e$ is the timeslot during which the edge begins. Let $t_e \equiv (p_e,s_e)$ such that $p_e$ is the phase in the schedule corresponding to $t_e$, and $s_e$ is the scale factor used during $t_e$. We wish to show that $\traversing(f(R,D),e) \leq 1$.

 We first use a greedy algorithm described in \Cref{sec:demand-inflation} to generate $\dprime$, a demand function such that for all $t$, $\dprime(t)$ has row and column sums exactly equal to $r$, and $\dprime(t)$ bounds $D(t)$ above. Due to the latter condition, it follows that $f(R,\dprime)$ bounds $f(R,D)$ above; thus $\traversing(f(R,\dprime),e) \geq \traversing(f(R,D),e)$.
 Henceforward, we focus on proving $\traversing(f(R,\dprime),e) \leq 1$.

 Valid paths in EBS include two components: the semi-path from the source node to an intermediate node, and the semi-path from the intermediate node to the destination node. We can therefore decompose the paths in $\traversing(f(R,\dprime),e)$ into two components as follows: first, we define $\rprime$, a routing protocol defined such that $\rprime_{a,b,t}(P)$ equals $1$ if $P$ is the semi-path from $(a,t)$ to $\timeline{b}$, and $0$ otherwise.
 Because EBS uses the same routing strategy for both source-intermediate semi-paths and intermediate-destination semi-paths, $\rprime$ is used for both components.
 Then, we introduce two demand functions: $\datob$ represents demand on semi-paths from origin nodes to intermediate nodes, while $\dbtoc$ represents demand on semi-paths from intermediate nodes to destination nodes. Note that for all physical edges $e$,
 \[ \traversing(f(R,\dprime),e)=\traversing(f(\rprime,\datob),e)+\traversing(f(\rprime,\dbtoc),e) . \]
To characterize $\datob$, note that regardless of source and destination, $R$ samples intermediate nodes uniformly. Therefore, for all $(t,a,b)\in \mathbb{Z} \times [N] \times [N]$,
\[ \datob(t,a,b) = \frac{1}{N} \sum_{c \in [N]} \dprime(t,a,c) = \frac{r}{N} \]
Similarly, because semi-paths from an intermediate node to the destination always commence exactly $T$ timeslots after the starting vertex, we can characterize $\dbtoc(t,b,c)$ as follows:
\[ \dbtoc(t,b,c) = \frac{1}{N} \sum_{a \in [N]} \dprime(t-T,a,c) = \frac{r}{N} \]
Note that $\datob = \dbtoc = \dall$, where $\dall$ is the uniform all-to-all demand function $\dall(t,a,b) = \frac{r}{N}$ for all $(t,a,b)\in \mathbb{Z} \times [N] \times [N]$. Therefore, $F(f(R,D),e) \leq 2 F(f(\rprime,\dall),e)$.

 \begin{clm} For all $e \in \ephys$, there are exactly $Tn^{\ell-1}$ triples $(t,a,b)$ such that the semi-path from $(a,t)$ to $\timeline{b}$ traverses $e$.
 \end{clm}
\begin{proof}[Proof of claim]
  Denote the endpoints of edge $e$ by $(i,t_e)$ and $(i + s \cdot \mathbf{e}_{p}, t_e+1)$.
 The semi-path of a triple $(t,a,b)$ traverses $e$ if and only if the semi-path first routes from $(a,t)$ to $(i,t_e)$, and $(b-a)_{p} = s$.

 Because semi-paths complete in $T$ timeslots, only semi-paths beginning in timeslots in the range $[t_e - T + 1 \dots t_e]$ could possibly reach node $(i,t_e)$ and traverse $e$. For every $t \in [t_e - T + 1 .. t_e]$, where $t \equiv (p_t,s_t)$, we can construct $n^{\ell-1}$ such triples as follows: First, we select $\bm{d}$, a vector representing the difference between $a$ and $b$ in the triple we will construct. To satisfy the second condition on $(t,a,b)$, we must set $\bm{d}_{p} = s$. However, the remaining $\ell-1$ indices of $\bm{d}$ can take on any of the $n$ possible values. Thus, there are $n^{\ell-1}$ possibilities for $\bm{d}$.

 For any semi-path $(t,a,b)$ such that $b-a = \bm{d}$, the timeslots in which a physical edge is traversed can be determined from $\bm{d}$. For any given timeslot $t^\prime \equiv (p^\prime,s^\prime)$ such that $t \leq t^\prime < t+T$, a physical edge is traversed if and only if $\bm{d}_{p^\prime} = s^\prime$. These are the edges that decrease the Hamming distance to $b$ by correctly setting coordinate $p$. We thus construct $a$ as follows: For every index $p$, if $(\bm{d}_p,p)$ is between $k_t$ and $k_e-1$ inclusive, we set $a_p$ = $\bm{i}_p - \bm{d}_p$. Otherwise, we set $a_p = \bm{i}_p$. Once we have constructed $a$, $b$ is simply $a+\bm{d}$. This choice of $a$ and $b$ ensures that by timeslot $t_e$, the semi-path from $(a,t)$ to $\timeline{b}$ reaches $\timeline{i}$.

 For each of the $T$ timeslots for which semi-paths originating in the given timeslot may traverse $e$, there are $n^{\ell-1}$ such semi-paths. This gives a total of $Tn^{\ell-1}$ semi-paths that traverse $e$ over all timeslots. Note that because each such semi-path has a unique $(t,\bm{d})$, none of the constructed semi-paths are double counted. In addition, because the $(t,\bm{d})$ pair determines the timeslots in which physical links are followed, and because there is only one physical link entering and leaving each node during each timeslot, there cannot be more than one choice of $a$ for a given $(t,\bm{d})$ pair such that the semi-path includes $(i,t_e)$. Because the $Tn^{\ell-1}$ count includes all possible choices of $\bm{d}$ for every timeslot, all semi-paths that traverse $e$ are accounted for.
\end{proof}
 Now we continue with the proof of \Cref{thm:shale-throughput}.
 Since exactly $Tn^{\ell-1}$ triples $(t,a,b)$ correspond to semi-paths that traverse $e$, and $\dall$ assigns $\frac{r}{N}$ flow to each semi-path, $F(f(\rprime,\dall),e) = \frac{r}{N} Tn^{\ell - 1} = \frac{r}{N} \ell (n-1) n^{\ell - 1}$. Thus:

 \begin{align*}
  F(f(R,D),e) &\leq 2 F(f(\rprime,\dall),e)
              =    2\frac{r}{N} \ell (n-1) n^{\ell-1}
              < 2\frac{r}{N} \ell n^{\ell}
              =    2\frac{r}{N} \ell (N^{1/\ell})^{\ell}
              =    2 r \ell
 \end{align*}

  When $r \leq \frac{1}{2\ell}$, for all physical edges $e$, $F(f(R,D,e)) \leq 1$. Thus, $f(R,D)$ is feasible.
\end{proof}

\subsection{Tightness of EBS Upper Bound}

\begin{lem}\label{lem:constant-factor-lb}
  For $0 < r \le \frac12$ let $h = \floor{\frac{1}{2r}}$
  and $\eps = h + 1 - \frac{1}{2r}.$
  The EBS design of order $h$ attains maximum latency
  at most $C \lowerbd(r,N)$, except when
\begin{equation*}
	\eps \geq 2\sqrt{\frac{2h}{\pi}}\left(\frac{2e}{C}\right)^{h} .
\end{equation*}
\end{lem}

\begin{proof}
\Cref{thm:counting-lb} and \Cref{thm:shale-informal} together show the following about the maximum latency of EBS compared to the maximum latency lower bound:
\begin{align*}
	L_{EBS} & \leq 2hN ^{1/h} \\
	\lowerbd(r,N) & \geq \frac{h}{e} (\eps N)^{1/h} \left( \frac{\sqrt{\frac{\pi h}{2}}}{4h} \right) ^{1/h}
\end{align*}
 Note that this interpretation of the maximum latency lower bound is taken from equation \eqref{eq:i-dont-know} in the proof of \Cref{thm:counting-lb}.

Suppose we wish to assert $L_{EBS} / \lowerbd(r,N) \le C$.
Given $C$ and $h$, we will derive the possible values of $\eps$ for which
this assertion holds.
\begin{align*}
	C & \geq \frac{ 2hN ^{1/h}}{\frac{h}{e} (\eps N)^{1/h} \left( \frac{\sqrt{\frac{\pi h}{2}}}{4h} \right) ^{1/h}}
	= \frac{2e}{\left(\frac{\eps \sqrt{\pi h /2}}{4h}\right)^{1/h}} \\
	\frac{\eps \sqrt{\pi h /2}}{4h} & \geq \left(\frac{2e}{C}\right)^{h} \\
	\eps & \geq 2\sqrt{\frac{2h}{\pi}}\left(\frac{2e}{C}\right)^{h} .
\end{align*}
\end{proof}
When $\eps$ falls outside this range, the maximum latency
of the EBS design is far from optimal. In the following sections
we present and analyze an ORN design which gives a tighter upper
bound when $\eps$ falls outside this range, in other words when $\eps < 2\sqrt{\frac{2h}{\pi}}\left(\frac{2e}{C}\right)^{h}$.

\subsection{Vandermonde Bases Scheme}

In order to provide a tight bound when $\eps$ is very small, we define a new family of ORN designs which we term the Vandermonde Bases Scheme (VBS). VBS is defined for values of $N$ which are perfect powers of prime numbers. We begin by providing some intuition behind the design of VBS.

For $h = \floor{\frac{1}{2r}}$ and $\eps = h + 1 - \frac{1}{2r}$, a small value of $\eps$ indicates that $r$ is slightly above $\frac{1}{2(h+1)}$. This indicates that the average number of physical hops in a path can be at most slightly below the even integer $2(h+1)$.
EBS is only able to achieve an average number of physical hops equal to an even integer as $N$ becomes sufficiently large. In small $\eps$ regions, the difference between the highest average number of physical hops theoretically capable of guaranteeing $r$ throughput and the average number of physical hops used by EBS approaches $2$.
This suggests that EBS achieves a throughput-latency tradeoff that favors throughput more than is necessary in these regions, penalizing latency too much to form a tight bound.
A more effective ORN design for these regions would use paths with $2(h+1)$ physical hops, but mix in sufficiently many paths with fewer physical hops to ensure that the average number of physical hops per path is at most $2(h+1 - \eps)$.

VBS achieves this by employing two routing strategies for semi-paths alongside each other. The first strategy, single-basis (SB) paths, resembles the semi-path routing used by EBS for $h^\prime = h+1$. The second strategy, hop-efficient (HE) paths, will rely on the fact that VBS's schedule regularly modifies the basis used to determine which nodes are connected to one another.
HE paths will consider edges beyond the current basis, enabling them to form semi-paths between nodes using only $h$ hops, even when this is not possible within a single basis.
The more future phases are considered, the more nodes can be connected by HE paths. This tuning provides a high granularity in the achieved tradeoff between throughput and latency, and enables a tight bound in regions where $\eps$ is small. It is interesting that the quantitative reasoning underlying this scheme is reminiscent of the proof of the Counting Lemma (\Cref{lem:counting-lem}), which similarly classifies paths into short paths and long paths and counts the number of destinations reachable by short paths.

We define VBS for $N=n^{h+1}$ such that $n$ is a prime number. The connection schedule and routing algorithm of VBS depend on a parameter $\delta$, which represents a target for the fraction of semi-paths that traverse HE paths. We later describe how to set $Q$, the number of future phases considered for HE path formation, such that the number of destinations reachable by HE paths is approximately $\delta N$.

\subsubsection{Connection Schedule}

Before describing the connection schedule of VBS, it is instructive to revisit the schedule of EBS.
EBS's schedule consists of $h^\prime$ phases. Each of these phases is defined based on an elementary basis vector $\mathbf{e}_p$, connecting each node $\bm{i}$ to nodes $\bm{i} + s \mathbf{e}_p$ for all possible nonzero scale factors $s$. VBS is defined similarly, except instead of elementary basis vectors, Vandermonde vectors (to be defined in the next paragraph of this section) are used to form the phases.
 In addition, rather than using a single basis, the VBS connection schedule is formed from a longer sequence of phases, with any set of $h+1$ adjacent phases corresponding to a basis.


For VBS, we assume the total number of nodes in the system is $N = n^{h+1}$ for some prime number $n$. As in EBS, each node $a$ is assigned a unique set of $h+1$ coordinates $(a_0,a_1,...,a_h)$, each ranging from $0$ to $n-1$. This maps each node to a unique element of $\mathbb{F}_n^{h+1}$. We identify each permutation $\pi_k$ of the connection schedule using a scale factor $s$, $1 \leq s < n$ and a phase number\footnote{The mnemonic is that $p$ stands for ``phase number'', not ``prime number''. We beg the forgiveness of readers who find it confusing that the size of the prime field is denoted by $n$, not $p$.} $p$, $0 \leq p < n$, such that $k = (n-1) p + s - 1$. Each phase $p$ is formed using the Vandermonde vector $\bm{v}(p) = (1,p,p^2,...,p^{h})$. This produces the connection schedule $\pi_{(n-1) p + s - 1}(\bm{i}) = \bm{i} + s \bm{v}(p)$.

\subsubsection{Routing Algorithm}

As with EBS, VBS's oblivious routing scheme is based around VLB. First, traffic is routed along a semi-path from the source to a random intermediate node in the network, and then traffic is routed along a second semi-path from the intermediate node to its final destination. As in EBS, the same algorithm is used to generate semi-paths in both stages of VLB. However, unlike in EBS, semi-paths are only defined starting at phase boundaries. Thus, the first step of a VBS path is to traverse up to $n-2$ virtual edges until a phase boundary is reached. Paths are then defined for a given $(q,a,b)$ triple, where $q = t / (n-1)$ for some timeslot $t$ at the beginning of a phase (hence $t$ is divisible by $n-1$). Following the initial virtual edges to reach a phase boundary, we concatenate the semi-path from the source to the intermediate node, followed by the semi-path from the intermediate node to the destination.

Depending on the current phase and the source-destination pair, we either route via a single-basis path or a hop-efficient path. The routing scheme always selects a hop-efficient semi-path when one is available, and otherwise it selects a single-basis path. We describe both path types below.

\paragraph*{Single-basis paths}

The single-basis path, or SB path, for a given $(q,a,b)$ is formed as follows: First, we define the distance vector $\bm{d} = b - a$, as well as the basis $Y = (v(q), v(q+1), ..., v(q+h))$. Note that the vectors in the basis $Y$ are those used to form the $h+1$ phases beginning with phase $q$. Then, we find $\bm{s} = Y^{-1} \bm{d}$. Over the next $h+1$ phases, for every timeslot $t^\prime \equiv (p^\prime,s^\prime)$, if $s^\prime = \bm{s}_{p^\prime}$, the physical edge is traversed. Otherwise, the virtual edge is traversed. This strategy corresponds to traversing $\bm{d}$ through its decomposition in basis $Y$, beginning at node $a$ and ending at node $b$.

Although this algorithm for SB paths completes within $h+1$ phases, following this virtual edges are traversed for a further $Q$ phases. This ensures that both SB and HE paths take $h+1+Q$ phases to complete. Note that it is possible for an SB path to have fewer than $h+1$ hops, although this becomes increasingly rare as $N$ grows without bound.

\paragraph*{Hop-efficient paths}

A hop-efficient path, or HE path, is formed as follows: First, for $h+1$ phases, only virtual edges are traversed. This ensures that the physical hops of HE and SB paths beginning during the same phase $q$ use disjoint sets of vectors (assuming $n > h + 1 + Q$), which simplifies later analysis. Following this initial buffer period, $h$ phases are selected out of the next $Q$ phases, and one physical hop is taken in each selected phase. During all other timeslots within the $Q$ phases, virtual hops are taken.

For a given starting phase $q$ and starting node $a$, there are $ {Q \choose h} (n-1)^h$ possible HE paths. Because there are a total of $N$ destinations reachable from $a$, we would like $\delta N$ destinations to be reachable by HE paths.
Ignoring for now the possibility of destinations reachable by multiple HE paths, we set $Q$ to the lowest integer value such that:

\begin{align*}
 {Q \choose h} (n-1)^h \geq \delta N \;\Longleftarrow \;{Q \choose h} \geq \delta n \\
\end{align*}

Note that for this value of $Q$, ${Q-1 \choose h} < \delta n$.
For some $(q,a,b)$, more than one HE path may exist. In this case, an arbitrary selection can be made between these multiple paths; the specific path chosen does not affect our analysis of VBS.

\subsection{Latency-Throughput Tradeoff of VBS}

\subsubsection{Latency}

A VBS path begins with at most $n-2$ virtual edges traversed until a phase boundary is reached. Following this, the first semi-path immediately begins, followed by the second semi-path. Because both SB and HE paths are defined to take $h+1+Q$ phases, the latency of a single semi-path is $(n-1)(h+1+Q)$. This gives a total maximum latency of $(n-2) + 2(n-1)(h+1+Q) = (n-1)(3+2h+2Q) - 1$ for VBS paths.

\subsubsection{Throughput}

\begin{lem}\label{thm:vbc-throughput}
 Let $R$ be the VBS routing scheme for a given $N$, $h$, and $\delta$, such that
 $\delta \leq \frac{1}{4 (h+1) (1+ \frac{1}{2h})^2}$.
 For all demand functions $D$
 requesting throughput at most $\frac{1}{2(h+1-\eps)}$, where
 $\eps = \frac{1}{4}\delta$,
 the flow $f(R,D)$ is feasible.
\end{lem}
\begin{proof}
 Consider an arbitrary demand function $D$ requesting throughput at most $r$, and consider an arbitrary physical edge $e \in \wphys$ from $(i,t_e)$ to $(j,t_e+1)$, where $t_e$ is the timeslot during which the edge begins. Let $t_e \equiv (p_e,s_e)$ such that $p_e$ is the phase in the schedule corresponding to $t_e$, and $s_e$ is the scale factor used during $t_e$. We wish to show that $\traversing(f(R,D),e) \leq 1$.

 As in our proof of the throughput of EBS (\Cref{thm:shale-throughput}), we begin by inflating $D$ into $\dprime$. Similarly, we define $\rprime$, the routing protocol for semi-paths, and we decompose $f(R,\dprime)$ into $f(\rprime,\datob)$ and $f(\rprime,\dbtoc)$. Note that because semi-paths begin only on phase boundaries, $\rprime$ in this case does not strictly follow our definition for an oblivious routing scheme. Instead, we define $\rprime_{a,b,q}$ using phases $q$, rather than timeslots $t$, for the domain. The path used for $\rprime_{a,b,q}$ begins during the first timeslot of phase $q$. This is reflective of the definitions for semi-paths in VBS.

 To generate $\datob$, note that $R$ first batches $(a,b,t)$ triples over the $n-1$ timeslots preceding an epoch boundary, before sampling intermediate nodes uniformly. Therefore, for all $(q,a,b)$

 \[ \datob(q,a,b) = \frac{1}{N} \sum_{t \in [n-1]}\sum_{c \in [N]} \dprime(q(n-1)-t,a,c) = \frac{(n-1)r}{N} \]

 Similarly, because semi-paths from an intermediate node to the destination always commence exactly $h+1+Q$ phases after the beginning of the first semi-path, we can define $\dbtoc(t,b,c)$ as follows:

 \[ \dbtoc(q,b,c) = \frac{1}{N} \sum_{t \in [n-1]}\sum_{c \in [N]} \dprime((q-h-1-Q)(n-1)-t,a,c) = \frac{(n-1)r}{N} \]

 Note that $\datob = \dbtoc = \dall$, where $\dall$ is the uniform all-to-all demand function $\dall(q,a,b) = \frac{(n-1)r}{N}$ for all $(q,a,b)\in \mathbb{Z} \times [N] \times [N]$. Therefore, $F(f(R,D),e) \leq 2 F(f(\rprime,\dall),e)$.

 To calculate $F(f(\rprime,\dall),e)$, we compute the number of $(q,a,b)$ triples that traverse edge $e$. We calculate this number as follows: First, we calculate $\#_{SB}$, which represents the number of $(q,a,b)$ triples that have an SB path that traverses edge $e$. Then, we calculate $\#_{missing}$, the number of such triples that have an HE path available (and thus do not traverse e). Finally, we determine $\#_{HE}$, the number of triples that traverse $e$ using an HE path. The total flow traversing edge $e$ is then $F(f(\rprime,\dall),e) = \frac{(n-1)r}{N}(\#_{SB} - \#_{missing} + \#_{HE})$.

 To find $\#_{SB}$, we use reasoning similar to that used in \Cref{thm:shale-throughput}. In order for a given $(q,a,b)$ to have an SB path that traverses edge $e$, the SB path for $(q,a,b)$ must reach node $(i,t)$, then traverse edge $e$. The only values of $q$ for which this is possible are those in the range $q_e - h \leq q \leq q_e$. For each of these $q$, we can generate $n^h$ distinct $(q,a,b)$ triples that have SB paths that traverse edge $e$ as follows. First, select an arbitrary $\bm{s}$ such that $s_{q_e - q} = s_e$. Then, set $a = \bm{i} - \Sigma_{q^\prime = q}^{q_e - 1} s_{q^\prime - q} v(q^\prime)$, and $b = a + \Sigma_{q^\prime = q}^{q+h} s_{q^\prime - q} v(q^\prime)$. In this case, $\bm{s}$ corresponds to a distance vector between $a$ and $b$, expressed in terms of the basis used for SB paths starting in phase $q$. Because of how $a$ is set, it is clear that the SB path for $(q,a,b)$ must traverse $(i,t)$. In addition, because $s_{q_e - q} = s_e$, the SB path will traverse edge $e$ instead of another edge during the same phase.

 For a given $q$, there are $n^h$ possible values for $\bm{s}$, because all but one of its $h+1$ elements can be set to any value in $[n]$. There are $(h+1)$ possible values for $q$, giving a total of $\#_{SB} = (h+1)n^h$

 To find $\#_{missing}$, we compare the distance vectors of $(q,a,b)$ triples that have SB paths which traverse $e$ with those of $(q,a,b)$ triples that have valid HE paths.
 Each vector found in the overlap between these two sets corresponds to one $(q,a,b)$ triple that contributes to $\#_{missing}$.
 To reason about the former set of vectors, we return to the construction of $\bm{s}$ used to find $\#_{SB}$.
 For a given starting phase $q$, each $\bm{s}$ such that $\bm{s}_{q_e - q} = s_e$ represents a distance vector that can traverse $e$, expressed in terms of the basis used for SB paths starting in phase $q$.
 We can construct this basis as $Y = (v(q), v(q+1), ..., v(q+h)$.
 For each $\bm{s}$, $\bm{d} = Y \bm{s}$ is the same distance vector expressed using the elementary basis.
 The range of possible distance vectors $\bm{d}$ reachable while traversing $e$ forms $D_e$, an $h$-dimensional affine subspace of $\mathbb{F}_n^{h+1}$ that is parallel to $W_e$, the linear subspace spanned by the set $Y \setminus \{v(q_e)\}$.

 Next, we consider which triples have valid HE paths.
 For a given starting phase $q$, there are $Q$ phases which are considered for forming HE paths.
 Let $I$ be a set of $h$ phase numbers chosen from these $Q$ phases, and let $V(I)$ be the linear subspace spanned by the vectors corresponding to the phase numbers in $I$.
 There are ${Q \choose h}$ ways of choosing such a set $I$.
 For each possible choice, $V(I)$ forms an $h$-dimensional linear subspace in $F_n^{h+1}$, corresponding to the distance vectors reachable via HE paths using the chosen phases. (Note that $V(I)$ must be $h$-dimensional because every $h$ distinct Vandermonde vectors are linearly independent.)
 Because $V(I)$ and $W_e$ are spanned by distinct sets of $h$  Vandermonde vectors, these linear subspaces are not equivalent, implying that $V(I)$ and $D_e$ are not parallel.
 Thus, $V(I) \cap D_e$ is an affine subspace with dimension $h-1$ and contains $n^{h-1}$ distance vectors.

 Some distance vectors lie in more than one such intersection.
 In order to avoid overcounting $\#_{missing}$, we must remove at least this many vectors from our count.
 Given two sets of $h$ chosen phase numbers $I$ and $J$, $V(I)$ and $V(J)$ form two different linear subspaces of $\mathbb{F}_n^{h+1}$.
 As linear subspaces, both $I$ and $J$ contain the zero vector, as does the $(h-1)$-dimensional $I \cap J$.
 $D_e$ does not contain the zero vector, so $D_e \cap I \cap J$ can only be $(h-2)$-dimensional, containing $n^{h-2}$ distance vectors.
 There are fewer than ${Q \choose h} ^2$ ways of choosing two distinct sets $I$ and $J$.

 Thus, for a given starting $q$, there are fewer than ${Q \choose h} n^{h-1} - {Q \choose h}^2 n^{h-2}$ distance vectors in the overlap between $D_e$ and the union of all possible $V(I)$.
 Because there are $h+1$ possibilities for the starting $q$, this gives
 the following lower bound for $\#_{missing}$:
 \begin{align*}
  \#_{missing} &>    (h+1) \left({Q \choose h} n^{h-1} - {Q \choose h}^2 n^{h-2}\right)\\
               &\geq (h+1) \left((\delta n) n^{h-1} - \left({Q-1 \choose h}\frac{Q}{Q-h}\right)^2 n^{h-2}\right)\\
               &>    (h+1) \left(\delta n^{h} - \left(\delta n\frac{Q}{Q-h}\right)^2 n^{h-2}\right)\\
               &=    (h+1) \left(\delta n^{h} - \delta^2 n^{h} \left(\frac{Q}{Q-h}\right)^2\right)\\
 \end{align*}

 To find $\#_{HE}$, note that a given $(q,a,b)$ can only traverse edge $e$ if $q_e - h - Q \leq q < q_e - h$, since $q_e$ must be in the set of $Q$ phases considered for HE paths for $(q,a,b)$. For a given $q$, we can construct an HE path by selecting $h-1$ additional phases from the $Q-1$ remaining phases, and then selecting one of the $n-1$ edges within that phase to traverse. Some of these paths may lead to the same destination, causing an overcount, but it is fine to overcount $\#_{HE}$ slightly.

 \begin{align*}
  \#_{HE} &\leq Q {Q-1 \choose h-1} (n-1)^{h-1} \\
          &=    Q {Q-1 \choose h} \frac{h}{Q-h} (n-1)^{h-1} \\
          &<    \delta n h \frac{Q}{Q-h} (n-1)^{h-1} \\
          &<    \delta h n^h \frac{Q}{Q-h}
 \end{align*}

 Now that we have found $\#_{SB}$, $\#_{missing}$, and $\#_{HE}$, we can finally bound $F(f(R,D),e)$:

 \begin{align*}
  F(f(R,D),e)           &\leq 2 F(f(\rprime,\dall),e) \\
                        &= 2\frac{(n-1)r}{N}(\#_{SB} - \#_{missing} + \#_{HE}) \\
                        &< 2\frac{(n-1)r}{N}\left((h+1)n^h - (h+1)\left(\delta n^h - \delta^2 n^h \left(\frac{Q}{Q-h}\right)^2\right) + h \delta n^h \frac{Q}{Q-h}\right)\\
                        &= 2\frac{(n-1)r}{N}(h+1)n^h\left(1 - \left(\delta - \delta^2 \left(\frac{Q}{Q-h}\right)^2\right) + \frac{h}{h+1} \delta \frac{Q}{Q-h}\right)\\
                        &< 2r(h+1) \left(1 - \delta \left(1 - \frac{h}{h+1}\frac{Q}{Q-h}\right) + \delta^2 \left(\frac{Q}{Q-h}\right)^2\right)
 \end{align*}

 For $Q \geq 2h^2-h$, $\frac{Q}{Q-h} \leq \frac{h + \frac{1}{2}}{h}$. This gives:
 \begin{align*}
  F(f(R,D),e)           &< 2r(h+1) \left(1 - \delta \left(1 - \frac{h}{h+1}\frac{h+\frac{1}{2}}{h}\right) + \delta^2 \left(\frac{h+\frac{1}{2}}{h}\right)^2\right)\\
                        &= 2r(h+1) \left(1 - \delta \left(1 - \frac{h+\frac{1}{2}}{h+1}\right) + \delta^2 \left(1 + \frac{1}{2h}\right)^2\right)\\
                        &= 2r(h+1) \left(1 - \frac{1}{2}\frac{1}{h+1}\delta + \delta^2 \left(1 + \frac{1}{2h}\right)^2\right)\\
                        &= \frac{1}{2(h+1-\eps)}2(h+1) \left(1 - \frac{1}{2}\frac{1}{h+1}\delta + \delta^2 \left(1 + \frac{1}{2h}\right)^2\right)\\
                        &= \frac{1}{h+1-\eps}\left(h + 1 - \frac{1}{2}\delta + (h+1)\delta^2 \left(1 + \frac{1}{2h}\right)^2\right)\\
                        &\leq \frac{1}{h+1-\eps}\left(h + 1 - \eps\right)\\
  F(f(R,D),e)           &< 1
 \end{align*}

 Note that because of how we set $\eps$ and restrict $\delta$,
 $\eps \leq \frac{1}{2}\delta - (h+1)\delta^2 (1 + \frac{1}{2h})^2$.
 Because the amount of flow traversing any physical edge $e$ is less than 1, the flow $f(R,D)$ is feasible.

\end{proof}

\subsection{Tightness of Upper Bound}

\begin{thm}\label{thm:tight_bounds}
	For all $r\in(0,1/2]$, there is a VBS design or an EBS design which guarantees throughput $r$ and uses maximum latency
	\begin{equation}
		L_{max} \leq 
    O(\lowerbd(r,N)) .
	\end{equation}
\end{thm}
\begin{proof}
The VBS design
of order $h$ with parameter $\delta$
 gives maximum latency $L \leq (h+1)(n-1) + Q(n-1)$ for $h = \floor{\frac{1}{2r}}$, ${Q\choose h} \geq \delta n$, as long as $\delta \leq \frac{1}{4(h+1)(1+\frac{1}{2h})^2}$. Let $\eps = h+1-\frac{1}{2r}$, and set $\delta = 4\eps$.

We chose $Q$ such that ${Q-1\choose h} < \delta n$ and ${Q \choose h} \ge \delta n$. Then ${Q\choose h} < \delta n \frac{Q}{Q-h} \leq \delta \frac{h + \frac{1}{2}}{h}$, due to $Q\geq 2h^2-h$. Hence $Q \leq h \left(\delta n \frac{h}{h+(1/2)}\right)^{1/h}$. We upper bound the max latency of VBS in the following way.

\begin{align*}
	L_{max} & \leq \max\{ (h+1)(n-1) + Q(n-1), (h+1)(n-1) + (2h^2-h)(n-1) \} \\
	& \leq 2(h+1)(n-1) + 2h^2(n-1) + h\left(4\eps n\frac{h+\frac{1}{2}}{h}\right)^{1/h}(n-1) \\
	& \leq 2(h+1)n + 2h^2n + hn(4\eps n)^{1/h}\left(\frac{2h+1}{2}\right)^{1/h} \\
	& \leq (h+1)[ 2N^{1/(h+1)} + h N^{1/(h+1)} + (4\eps N)^{1/h} \left(\frac{2h+1}{2}\right)^{1/h} ] \\
	& \leq O(h[ h N^{1/(h+1)} + (\eps N)^{1/h} ])
\end{align*}

For sufficiently large $N$ (determined by $\eps$ and $h$, both functions of $r$), the second term will dominate. Thus, for large N:

\begin{equation*}
	L_{max} \leq O\left(h \left[(\eps N)^{1/h} + N^{1/(h+1)}\right]\right)
   = O \left( \lowerbd(r,N) \right).
\end{equation*}

By \Cref{thm:vbc-throughput}, VBS only gives a tight latency bound when $4\eps = \delta \leq \frac{1}{4 (h+1) (1+ \frac{1}{2h})^2}$. When $\eps$ is greater than this value, we use EBS instead. By \Cref{lem:constant-factor-lb}, EBS gives a factor $C$ tight bound when $\eps > 2\sqrt{\frac{2h}{\pi}}\left(\frac{2e}{C}\right)^{h}$. We check to make sure that there exists a constant $C$ which works for all $\eps > \frac{1}{4}\cdot\frac{1}{4 (h+1) (1+ \frac{1}{2h})^2}$

\begin{align*}
 	2\sqrt{\frac{2h}{\pi}}\left(\frac{2e}{C}\right)^{h} & \leq \frac{1}{4}\cdot\frac{1}{4 (h+1) \left(1+ \frac{1}{2h}\right)^2} \\
 	\frac{2e}{C} \left(2\sqrt{\frac{2h}{\pi}}\right)^{1/h} & \leq \left(\frac{1}{16 (h+1) \left(1+ \frac{1}{2h}\right)^2}\right) ^ {1/h} \\
 	C & \geq 2e\left(2\sqrt{\frac{2h}{\pi}}\right)^{1/h} \left(16 (h+1) \left(1+ \frac{1}{2h}\right)^2\right) ^ {1/h} \\
 	C & \geq O\left(\sqrt{h}^{1/h} \left((h+1)\left(\frac{2h+1}{2h}\right)^2\right)^{1/h}\right)
 	= O(1)
\end{align*}

Since there exists such a factor $C$, the following holds for EBS in the regions of interest.

\begin{equation*}
	L_{max} \leq O\left(h \left[(\eps N)^{1/h} + N^{1/(h+1)}\right]\right)
   = O \left( \lowerbd(r,N) \right)
\end{equation*}

\end{proof}

\subsection{Showing the Upper Bound for $d>1$}
\label{sec:d-regular-upper-bound}


Recall from \Cref{sec:d-greater-than-one} that an upper bound for 1-regular designs will only imply a similar upper bound for $d$-regular designs if we can ensure that the routing scheme does not route flow paths on multiple edges in the same ``unrolled''  segment of the 1-degree virtual topology.
EBS and VBS always route flow on paths which use at most 1 edge from each phase, where a phase constitutes $(n-1)$ timeslots. Trivially, if $d$ divides $(n-1)$, then these constructions already have the property we need. However, even if $d$ does not divide $(n-1)$, as long as $d<n-1$, we can modify EBS and VBS as follows.

We change the connection schedule to iterate through each phase twice before moving on to the next. So for VBS, $\pi_{(n-1) p + s - 1}(\bm{i}) = \bm{i} + s \bm{v}(\floor{p/2})$. We also change the definition of single-basis and hop-efficient paths to use exclusively even-numbered phases or exclusively odd-numbered phases, depending on whether the next phase starts after the request originates.
With this modification, single-basis and hop-efficient paths always use physical edges that occur at least $(n-1)$ timeslots apart from each other. Therefore, in the ``rolled up'' virtual topology, our flow paths will always use at most one physical edge per timeslot. This at most doubles the maximum latency, and does not affect throughput.

\section{Conclusion and Open Questions} \label{sec:open-questions}

In this paper we introduced a mathematical model
of oblivious reconfigurable network design and
investigated the optimal latency attainable for
designs satisfying any given throughput
guarantee, $r$.
We proved that the best maximum latency achievable
is $\Omega(\lowerbd(r,N))$, for 
$\lowerbd(r,N) = h \left( N^{1/(h+1)} + (\eps N)^{1/h} \right)$.
We also present two ORN designs, EBS and VBS.
For every constant $r$, we show there exist infinitely many $N$ for which
either EBS or VBS achieves a maximum latency of
$\mathcal{O}(\lowerbd(r,N))$.

Our investigation of the throughput-latency tradeoff
for ORN designs affords numerous opportunities for
follow-up work. In this section we sketch some of the
most appealing future directions.

%

\subsection{Universal connection schedules}
\label{opq:universal}

EBS and VBS both use connection schedules
tuned to the specific
throughput rate, $r$, that they aim to guarantee.
Is there a single connection schedule that permits
achieving the Pareto-optimal latency for a large range of
of $r$, or perhaps even for every value of $r$, merely
by varying the routing scheme?

We conjecture that the following connection schedule,
inspired by \cite{tremel}, supports ORN designs that
are Pareto-optimal with respect to the tradeoff between
worst-case throughput and {\em average} latency,
for every value of $r$, when $N$ is a
prime power. Let $\mathbb{F}$ denote the finite field
with $N$ elements, and let $x$ denote a primitive root
in $\mathbb{F}$. Define the sequence of permutations
$\pi_0, \pi_1, \ldots$ by specifying that $\pi_k(i) = i + x^k$
for all $i \in \mathbb{F}, \, k \in \mathbb{N}$. We have
experimented with this family of connection schedules
when $\mathbb{F}$ is a prime field and 2 is a primitive
root, for values of $N$ ranging from 11 up to around 300.
We numerically verified that in all cases we tested,
for each value of $r$ ranging from $\frac12$ down to
roughly $\frac{1}{\log n}$, there is an oblivious routing scheme
guaranteeing throughput $r$, whose average latency
is within a constant factor of matching our lower bound.
In fact, the average latency in most cases that we tested
was moderately less than EBS's. However, thus far we have not
succeeded in proving that this pattern persists for infinitely
many $N$.

\subsection{Bridging the gap between theory and practice}
\label{opq:theory-practice}

Our model of ORNs incorporates idealized assumptions that
gloss over important details that affect the performance
of ORNs in practice. A more realistic model would not equate
expected congestion with actual congestion. This would
necessitate grappling with the issues of queueing and
congestion control. It also opens the Pandora's box of
non-oblivious routing, since a frame that was intended
to be transmitted on link $(u,v)$ but finds that link
blocked due to congestion must either be transmitted in
a different timeslot, or on a different link in the same
timeslot, and in either case the frame's path in the
virtual topology differs from the intended one. An
appealing middle ground between fully centralized
control (as in classical models of circuit-switched
networks) and a fully oblivious model (as in our paper)
could be a network design with a fully oblivious
connection schedule coupled with a partially-adaptive
routing scheme based on local information such as
queue lengths at the transmitting and receiving
nodes.

Our model also fails to account for
(possibly heterogeneous) propagation delays, due to our
assumption that each link of the virtual topology corresponds
to exactly one timeslot regardless of where its endpoints
are situated.
The model could be enhanced to take propagation delay
into account by adjusting the virtual topology. Rather than
connecting physical edges from $(i,s)$ to $(j,s+1)$,
they could instead connect to $(j,s+d_{ij})$, where $d_{ij}$ is
a whole number representing the propagation delay from $i$ to $j$
in units of timeslots. As in our basic model, nodes of the virtual
topology in this enhanced model would be constrained to belong to at
most one incoming and at most one outgoing physical edge, though
if $d_{ij}$ varies with $i$ and $j$ then the set of physical edges
would no longer be described by a sequence of permutations.

\subsection{Supporting multiple traffic classes}
\label{opq:traffic-classes}

In this paper we sought to optimize the worst-case
latency guarantee for network designs that guarantee a
specified rate of throughput. In practice, flows co-existing
on a network can differ markedly in their latency
sensitivity. Can EBS, VBS, or other ORN designs be adapted
to offer users a menu of options targeting different points
on the latency-throughput tradeoff curve? What guarantees
can such network designs simultaneously provide to the
different classes of traffic they serve?

\endgroup

\section*{Acknowledgements}

This work was supported in part by NSF grants
CCF-1512964, CSR-1704742,
and CNS-2047283, a Google faculty research scholar award,
and a Sloan fellowship.

\bibliographystyle{alpha}
\bibliography{biblio}{}

\newpage
\appendix

\section{A general upper bound on achievable throughput in ORNs} \label{app:general-throughput}

The use of Valiant load balancing inflates path lengths by a factor of 2, which reduces throughput by a factor of 2. It turns out that this factor-2 loss is unavoidable for ORN designs. It is instructive to present a proof that no ORN design can sustain throughput greater than $\frac{1}{2} + o(1)$, even if latency is allowed to be unbounded. 

Consider the following:
let $\sigma$ denote a random permutation of the nodes, and consider a workload $D$ in which every node $a$ sends flow to destination $\sigma(a)$ at rate $r$. We will say a ``direct link'' is one whose endpoints are $a$ and $\sigma(a)$ for some node $a$, and a ``spraying link'' is any other physical link. Define the inflated cost of a link to be 2 if it is a direct link and 1 if it is a spraying link.

This ensures that the inflated cost of {\em every} routing path from $a$ to $\sigma(a)$ is at least 2, regardless of whether it is a direct or indirect path. 
Therefore, when an ORN design is used to route workload $D$ over a span of $T$ timeslots, the total inflated cost of the links used, weighted by their flow rates, is at least $2 r N T$. (In each of $T$ timeslots, each of $N$ nodes sends flow at rate $r$ on a routing path of inflated cost at least 2.) 
On the other hand, the {\em expected} total inflated cost of all physical edges in the virtual topology is $\left( 1 + \frac{1}{N-1} \right) NT$. This is because the virtual topology contains $NT$ physical edges, and the expected inflated cost of each $e$ is $1 + \frac{1}{N-1}$, accounting for the $\frac{1}{N-1}$ probability that the random permutation $\sigma$ leads us to label $e$ as a direct link and inflate its cost from 1 to 2.

If an ORN design sustains throughput $r$, then the flow rate on any physical edge in the virtual topology when routing workload $D$ is at most 1, and consequently the total inflated cost of all the physical edges used, weighted by their flow rates, is bounded above by the combined inflated cost of all the physical edges in the virtual topology. Hence $2r NT \le \left( 1 + \frac{1}{N-1} \right) NT$ and $r \le \frac12 + \frac{1}{2(N - 1)}$. This upper bound on throughput converges to $1/2$ as $N \to \infty$.

\section{Demand function inflation} \label{sec:demand-inflation}

Suppose we have a periodic demand function $D$ such that for all $t \in \mathbb{Z}$, $D(t)$ has row and column sums bounded above by $r$. Here, we present a greedy algorithm for inflating $D$ to produce $\dprime$, a demand function such that for all $t \in \mathbb{Z}$, $\dprime(t)$ has row and column sums exactly equal to $r$, and $\dprime(t)$ bounds $D(t)$ above:

\begin{algorithm}
 \For{$t$ in $\mathbb{Z}$}{
   $\dprime(t) = D(t)$\\
         \While{$\exists x \in [N] : \sum_{y \in [N]} \dprime(t,x,y) < r$}{
            Find the lowest $x \in [N]$ such that $\sum_{y \in [N]} \dprime(t,x,y) < r$\\
               Find the lowest $y \in [N]$ such that $\sum_{x \in [N]} \dprime(t,x,y) < r$\\
               Increase $\dprime(t,x,y)$ by $min(r - \sum_{z \in [N]} \dprime(t,x,z), r - \sum_{z \in [N]} \dprime(s,z,y))$
         }
}
\end{algorithm}

For all $t \in \mathbb{Z}$, because cells in $\dprime(t)$ are only ever increased, it should be clear that $\dprime(t)$ bounds $D(t)$ above.

To show that the row and column sums of $\dprime(t)$ all exactly equal $r$, first note that no cell has its value increased in a way that would cause a row or column sum to exceed $r$. Next, note that if the algorithm terminates successfully, all row sums of $\dprime(t)$ are equal to $r$. This implies that the sum of all cells in $\dprime(t)$ is $Nr$. Assume there exists a column sum less than $r$. Even if all column sums equal $r$, this leads to a contradiction, as the total sum of all cells must be less than $Nr$. Therefore, all column sums must equal $r$ as well.

The only step in the algorithm that does not trivially succeed is finding the lowest column $y$ whose column sum is less than $r$. We show that this step must succeed through contradiction: Assume that this step fails because there is no column sum less than $r$. Because no column sum is increased to be greater than $r$, it follows that all column sums must equal $r$. Due to a similar argument as the previous paragraph, all row sums must equal $r$. However, if all row sums equal $r$, the algorithm should have already moved on to the next $t$, which is a contradiction. Therefore, the algorithm terminates successfully.

\end{document}